\documentclass[a4paper,USenglish,cleveref,autoref]{lipics-v2021}

\hideLIPIcs  

\nolinenumbers

\usepackage{style}
\title{Safe Exploration of Arbitrary Dynamic Dangerous Networks}
\author{Caterina Feletti}{School of Computer Science, Carleton University, Ottawa, Canada {\and} School of Electrical Engineering and Computer Science, University of Ottawa, Canada}{caterinafeletti@cunet.carleton.ca}{https://orcid.org/0009-0004-1813-8056}{}
\author{Paola Flocchini}{School of Electrical Engineering and Computer Science, University of Ottawa, Canada}{paola.flocchini@uottawa.ca}{https://orcid.org/0000-0003-3584-5727}{}
\author{Giuseppe Prencipe}{Department of Computer Science, Università di Pisa, Italy}{giuseppe.prencipe@unipi.it}{https://orcid.org/0000-0001-5646-7388}{}
\author{Nicola Santoro}{School of Computer Science, Carleton University, Canada}{santoro@scs.carleton.ca}{https://orcid.org/0000-0002-7954-3918}{}
\authorrunning{C. Feletti, P. Flocchini, G. Prencipe, and N. Santoro}
\Copyright{Caterina Feletti, Paola Flocchini, Giuseppe Prencipe, and Nicola Santoro}

\begin{document}

\maketitle
    \begin{abstract}
    Given a team of agents on the nodes of a graph-based network, the exploration problem requires each node to be visited by at least one agent.
    In the classical distributed setting of static networks, agents do not know the topology of the network; in the more recently investigated setting of dynamic networks, they may have prior knowledge about the graph class (e.g., trees, tori, rings) or other parameters (e.g., number of nodes).
    Gotoh \emph{et al.} (2021) are the first to study the exploration problem of \emph{arbitrary dynamic networks}, under the necessary minimal assumption that any two nodes will be connected by a temporal path infinitely often (\emph{temporal connectivity} assumption). 

    In this paper, we extend their study of exploration under temporal connectivity by considering \emph{dangerous} dynamic networks, i.e., containing possibly one or more black holes. Whenever an agent enters a black hole, it will be trapped forever. 
    The \emph{safe exploration} problem requires a team to explore all the safe nodes, ensuring that at least one agent will never be trapped in a black hole.
    We first prove that, given the necessary (and sufficient) number of agents, a team of oblivious agents can \emph{perpetually} explore the safe nodes without any prior knowledge of the network or the team, without agent or node IDs, under semi-synchronous schedulers. 
    Then, we provide an algorithm that enables agents to safely explore the network and \emph{terminate}.
    In this case, agents are equipped with unique IDs and persistent memory, and they know the number of safe nodes.
    Yet, in both cases, we prove that it is impossible for a team of agents on dynamic networks to correctly mark only the ports leading to black holes.
\end{abstract}

    \section{Introduction}
Exploration of networks is a fundamental problem in the field of distributed computing by mobile agents \cite{Das19_BOOK}.
Indeed, the possibility for a team of agents---such as mobile sensors, software agents, robots---to visit all the nodes of a network is a key preliminary subroutine in various contexts: e.g., monitoring node status, maintaining the network, patrolling, and mapping the topology of an unknown network.
Unsurprisingly, the exploration problem has been broadly studied from both theoretical and practical perspectives. 
Notably, the theoretical research has focused on algorithmic strategies and computational limits of different models of graph-based networked systems, which can vary both for the capabilities of mobile agents and for the network features.
First and foremost, the underlying graph of a network must be connected in order to be explored.
Besides this trivial assumption, much attention has been devoted to the \emph{memory and communication} capabilities of the agents (through pebbles, tokens, whiteboard, personal memory, zero-memory), their \emph{synchronization} (fully synchronous {\fsynch}, semi-synchronous {\ssynch}, and asynchronous {\asynch}), the \emph{graph class} of the network (e.g., arbitrary finite graphs \cite{BockenhauerFUW23,GotohFMS21}, cactus graphs \cite{ShimoyamaSKM22}, trees \cite{BockenhauerFUW23,PattanayakP24}, infinite graphs \cite{PattanayakP24}), the presence/absence of \emph{IDs} or local labels (on agents, on nodes, on the ports of the nodes), and the \emph{a priori knowledge} the agents have about the network (e.g., graph topology \cite{ErlebachKLSS19,ErlebachS18}, graph class\footnote{The \emph{topology} of a graph $G$ defines set of vertices $V(G)$ and edges $E(G)$; graphs may belong to different \emph{classes} according to their structural properties (e.g., planar graphs, trees, tori, hypercubes).}, number of nodes, unknown graph \cite{GotohFMS21}) and on the other agents (e.g., how many they are, their initial positions).
Deviating slightly from the traditional distributed setting, some works have investigated these problems for systems composed of a single agent, thus focusing on algorithmic strategies for individual graph navigation rather than agent cooperation \cite{BockenhauerFUW23, DevismesDL26,PattanayakP24}.

The (im)possibility of exploring a network by a team of agents depends not only on the assumptions of the model, but also on the specific type of exploration problem, which is defined by the \emph{task} to be accomplished and the \emph{initial configuration} from which the agents start the exploration.
Agents may be assumed to start from the same node, called \emph{home base}, from distinct nodes, or from an arbitrary arrangement possibly containing multiplicities (i.e., agents on the same node).
Also, the number of agents may play a crucial role in the problem's solvability.
According to the task, the \probfont{Perpetual Exploration} problem requires each node to be visited infinitely often; instead, the \probfont{Finite Exploration} (a.k.a. \emph{exploration with termination}) requires the agents to stop moving after each node has been visited at least once.
Further termination conditions may ask the agents to become aware of the exploration termination (\emph{explicit termination} \cite{LunaDFS20}), or, more specifically, to gather at the same node, possibly the initial home base, and terminate. 
In its most general variant, the \probfont{Exploration} problem simply asks that each node will eventually be visited, without further requests (e.g., in \cite{PattanayakP24} for infinite graphs).
 Thus, depending on the exploration variant, the literature has provided algorithmic solutions including interesting subroutines for mobile agents: e.g., \probfont{Gathering} if agents must gather at the same node, or \probfont{Map Construction} if the agents have to construct the map of the topology of the unknown network \cite{DasFKNS07}.

The majority of the work assumes agents operate on \emph{static graphs}, where the edges between the nodes are always present and fixed.
Although reasonable, this assumption does not capture the manifold of realistic scenarios where links between nodes may be temporarily or permanently unavailable due to failures, downtime, or link reconfigurations.
For this reason, \emph{temporal graphs} have been used to investigate how the possible lack of edges affects the capability of agents to explore the network \cite{Erlebach0K15, ErlebachKLSS19,ErlebachS18,  GotohFMS21}.
An adversary decides at any discrete time step $t$ the \emph{snapshot} of the temporal graph, i.e., which edges are actually present at time $t$.
An edge can stop appearing forever (\emph{transient} edges) or appear infinitely often (\emph{recurrent} edges) after periodic or unpredictable times.
However, the power of the adversary must be restricted in order to avoid any non-trivial problem from becoming unsolvable.
In this regard, the minimal property that the adversary must guarantee is \emph{temporal connectivity}, i.e., starting from any time $t$, all the pairs of nodes will be eventually connected by a temporal path, namely a \emph{journey}.
Yet, the literature has often considered stronger assumptions: \emph{periodic connectivity} requires that snapshots appear with a periodicity $p>1$, \emph{$T$-interval connectivity} for $T\geq 1$ assumes that for every window of $T$ time steps, a connected spanning subgraph persists for all the snapshots during this period \cite{IlcinkasW18,KuhnLO10}.
A special case is the 1-interval connectivity, which thus claims that each snapshot must be connected.
A further restriction on the adversary’s power is given by the \emph{$\ell$-bounded assumption}, which requires that at most $\ell$ edges may be absent from any snapshot \cite{GotohFMS21}.

Another type of network failure or adversary attack to be considered is the possible presence of dangerous nodes, commonly called \emph{black holes} after their ability to trap and destroy every agent that enters them.
Networks with black holes have been widely studied in the literature with regard to the \probfont{Black Hole Search} (\BHS) primitive \cite{BileskiM26_arxiv,DobrevFPS06,KaurS26,KaurSMM25_SSS,KaurSMM25,MarkouS19}: given a dangerous network (i.e., with a black hole), the agents must search for its location.

This paper aims to study the (perpetual and finite) exploration problem, under the combined assumptions of temporal graphs and the presence of black holes: hence the name \emph{safe exploration of dynamic dangerous networks}.

\subsection{Related work and contribution}
\paragraph*{Background: exploration of dynamic networks}
Besides the traditional aspects, the key factors in the study of the exploration problems for temporal networks are \emph{(i)} the {class} of the underlying graph, \emph{(ii)} the {connectivity assumption} of the temporal graph, and \emph{(iii)} the {a priori knowledge} that agents have of the network {topology} and the occurrences of the edges.
Indeed, some exploration problems turn out to be unsolvable under certain dynamic settings, even by increasing the power of the agents \cite{LunaDFS20}.
Most of the work has been done considering only specific graph classes (e.g., dynamic tori, trees, cactus, rings), under the strong 1-interval or T-interval connectivity assumptions \cite{Erlebach0K15,ErlebachKLSS19,ErlebachS18,GotohFMS21,GotohSOKM18,MandalMM23,MichailS16,SaRMMoSh26}.
Among these works, both known \cite{Erlebach0K15,ErlebachKLSS19,ErlebachS18,MichailS16,SaRMMoSh26}  and unknown \cite{GotohFMS21,GotohSOKM18} topologies have been considered, provided they are finite. 
In some works, agents may have partial knowledge of the topology, e.g., in terms of bounds on the graph size \cite{LunaDFS20}.

\paragraph*{Related work: unknown and arbitrary dynamic network}
Gotoh \emph{et al.} \cite{GotohFMS21} are the first to consider the exploration of unknown finite\footnote{In the remainder of the paper, we will always consider finite graphs; thus, we will omit this specification.} temporal networks whose graph class is \emph{arbitrary}.
In particular, they mainly consider the perpetual exploration problem under temporally connected graphs and 1-bounded 1-interval graphs.
Under these settings, they analyze the number of agents necessary and sufficient for the exploration, and provide algorithmic solutions under $\ssynch$ schedulers (namely, where an arbitrary subset of agents is activated at each round) and $\fsynch$ schedulers (namely, where all agents are activated at each round).
Note that, unlike most of the existing literature, their first algorithm works under the temporal connectivity assumption, namely the minimal (i.e., less restrictive) connectivity assumption that a temporal graph must hold so that a non-trivial problem, including exploration, can be solved.
Despite this very limited model, the authors propose an elegant algorithm solving perpetual exploration with oblivious agents (no personal memory) starting from arbitrary nodes\footnote{This setting is also called the \emph{scattered} setting.}, no IDs on agents or nodes, under {\ssynch}.
Each node contains a \emph{whiteboard}, i.e., a local persistent memory, initially blank, whose mutually exclusive read/write access is granted to any agent when activated in the node.
The edge ports are labeled with local, permanent, and arbitrary IDs; agents know from which port they have just entered a node.
The algorithm exploits the well-known \emph{collaborative rotor–router} mechanism \cite{DereniowskiKPU14,KosowskiP19}; this mechanism was developed as a
deterministic counterpart to random walks on graphs by multiple agents, extending the original single-agent use developed by Fraenkel \cite{Fraenkel70}.
Basically, on each node, the port labels are ordered, and a pointer is maintained so that it points to the next port to be visited, following a cyclic order. As soon as an agent decides to move along an edge, it increases the value of the pointer.
The authors in \cite{GotohFMS21} prove that, provided a sufficient number of agents, the collaborative rotor-router mechanism enables agents to solve perpetual exploration even if edges can unpredictably disappear for an unbounded, even infinite, time.

\paragraph*{Contribution: safe exploration}
In this paper, we extend the study on exploration of unknown temporal graphs by assuming that the network is \emph{dangerous}, i.e., with possibly one or more black holes.
Now, the challenge gets harder: not only do agents have no knowledge about the graph, and the presence of the edges is subject to an adversary, but agents must avoid being trapped in black holes before accomplishing the task. 

We refer to these versions of the problem as \emph{safe exploration}: a team is required to (perpetually or finitely) explore all the safe nodes, ensuring that at least one agent will never be trapped in a black hole.
To this aim, concurrently with the exploration task, agents try to properly mark, on the whiteboard of each node, the ports as either {\SAFE} or {\DANGEROUS}: the final \emph{port marking} will be used by the remaining agents to continue the exploration without falling into the black holes. Indeed, port marking is a task of independent interest.

As in \cite{GotohFMS21}, we assume that agents start from arbitrary nodes and that they are activated by {\ssynch} schedulers.
Under these assumptions, the classical \emph{cautious walk} technique for the {\BHS} problem---firstly described in \cite{DobrevFPS07}---cannot correctly identify all edges leading to black holes.
In fact, the classical cautious walk assumes the static {\fsynch} setting where agents start from the same node, and makes agents proceed as a group as follows: one agent plays the role of the \emph{probe} and explores an edge, while the others play the role of the \emph{guards} and wait for the return of the probe; if the probe does not return in the next round, the guards mark the port as dangerous. 
If the network is dynamic, the cautious walk requires two probes to verify a single port \cite{KaurSMM25}.

However, starting from a scattered configuration, there is no guarantee that two agents will gather at the same node before starting to probe the adjacent nodes.
To address this issue, in the recent work \cite{KaurSMM25_SSS} on {\BHS} in 1-bounded 1-interval connected networks, the authors adopt the \emph{individual cautious walk}; this version of the cautious walk works for scattered agents, but it assumes {\fsynch} schedulers and uses two agents as probes for each port.

The technique in \cite{KaurSMM25_SSS}---recently improved in \cite{BileskiM26_arxiv} to reduce the number of required agents---becomes ineffective under {\ssynch}: the guard agent may wait for the probe one for an unpredictable number of rounds, thus making it impossible to distinguish from an idle or a trapped agent.
We therefore formally prove that, under our setting, the problem of correctly marking all the ports is unsolvable. 
We therefore introduce other weaker variants of the port marking problem. 
Some of these variants can be solved in our setting and serve as subroutines for our agents in the safe exploration of dynamic networks.

Firstly, we focus on the \probfont{Safe Perpetual Exploration} ($\SEperp$) problem: as in \cite{GotohFMS21}, we assume oblivious and anonymous agents without any prior knowledge of the network, which is anonymous and arbitrary, under the temporal connectivity assumption.
In this very limited setting, we provide an algorithm that adopts the collaborative rotor-router mechanism and solves $\SEperp$ provided a sufficient number of agents.

Then, we consider the \probfont{Safe Finite Exploration} ($\SEfini$) problem: we adapt our first algorithm in order to make agents aware of the exploration termination, propagate the information along the network, and thus terminate.
To this aim, we assume agents are equipped with a unique ID and a persistent memory (\emph{notebook}). The only prior knowledge they have about the network is the number of safe nodes.

{Compared with existing studies, we highlight the key advances of our work:}
\begin{itemize} 
    \item \textbf{Problem combination.} Existing studies address the problem of exploration or the search for black holes independently. 
    Instead, our work combines the two problems, hence the name \emph{safe exploration}.
    Beyond its intrinsic interest, this approach is in line with the broader aim to investigate the complexity of combining two distinct problems in a single one, thereby considering more elaborate scenarios within the area of dynamic networks;
    \item\textbf{Zero or multiple black holes.} In both our algorithms, agents do not know a priori the number of black holes, which may be 0 up to $n-1$, where $n$ is the number of nodes.
    To the best of our knowledge, the literature about dangerous networks and {\BHS} assumes the presence of exactly one black hole, with the only exception of our work and \cite{ChalopinDS07};
    \item \textbf{More restricted assumptions:} The existing algorithms for {\BHS} for arbitrary dynamic networks work under the $\{1,f\}$-bounded\footnote{They consider dynamic graphs where the number of missing edges at each time can be at most 1 \cite{KaurSMM25_SSS,KaurSMM25} or $f>1$ \cite{KaurSMM25}.} 1-interval connectivity assumption, which is stronger than our temporal connectivity assumption; moreover, these algorithms terminate as soon as one agent has found \emph{one} of the ports leading to the (unique) black hole of the network \cite{KaurSMM25_SSS, KaurSMM25}.
    In contrast, here we investigate the computational power of a team of agents in marking \emph{all} the ports leading to the (zero or multiple) black holes, under the (minimal) temporal connectivity assumption.
    Moreover, most previous work on {\BHS} assumes fully synchronous agents; we consider the more adversarial case of {\ssynch}, where a scheduler chooses which agents are active at each round.
    \item \textbf{Sub-problem(s).} This work formally introduces and analyses the sub-problem \probfont{Port Marking}---i.e., mark the ports as {\SAFE} or {\DANGEROUS}---in four different versions (namely, \emph{correct}, \emph{highly reliable}, \emph{reliable}, and \emph{weakly reliable}).
    Beyond its independent interest, this distinction allows for comparing the computational power of agents in the context of the safe exploration problem under different settings.
    For example, we will prove that the impossibility of achieving correct port marking under minimal assumptions does not prevent agents from safely exploring a network; yet, we will prove that agents can achieve a highly reliable (weakly reliable, resp.) port marking while solving $\SEperp$ ($\SEfini$, resp.).
\end{itemize}
\Cref{tab:related_works} compares our contribution with the related works. Due to lack of space, some proofs, figures, and tables are provided in the appendix.
\begin{table}[t]
\centering
\resizebox{\textwidth}{!}{  
\renewcommand{\arraystretch}{1.3}{
    \begin{tabular}{|l|ccccc|cc| c|}
        \hline
        \textbf{Problem} & \textbf{Network} & $|B|$ & \textbf{Synch} & \textbf{Initial} & \textbf{Params} & \textbf{NoteB} & \textbf{WhiteB} & \textbf{AgentId}\\
        \hline
        \hline
        \cite{KaurSMM25} BHS & $\{1,f\}$-bounded 1-interval & 1& {\fsynch} & Home& \crossmark& $\log{n}$ & $\log{\Delta}$ & \greencheck\\ 
        \hline
        \cite{KaurSMM25_SSS} BHS  & 1-bounded 1-interval& 1&{\fsynch}& Any& \crossmark& $\log{n}$& $\log{n}$& \greencheck\\
        \hline
        \hline
        \cite{GotohFMS21} Perp. Explor.& Temp. conn. & 0& {\ssynch}& Any & \crossmark & \crossmark  & $\log{\Delta}$ & \crossmark  \\
        \hline
        \rowcolor{yellow!30}\textbf{Safe Perp. Explor.} &Temp. conn. & $[0,n)$ & {\ssynch} & Any & \crossmark & \crossmark & $\Delta$ & \crossmark\\
        \hline
        \rowcolor{yellow!30}\textbf{Safe Finite Explor.}  &Temp. conn. & $[0,n)$& {\ssynch} & Any & $s$ & $n\log{n}$ & $n\log{n}$ &\greencheck\\

        \hline
    \end{tabular}
}
} 
\caption{Related works vs. our contribution (yellow lines). $|B|$ is the number of black holes, $s$ is the number of safe nodes. Notebook and whiteboard sizes are expressed in O-notation w.r.t. $n$ (number of nodes), and $\Delta$ (degree of the network).}
\label{tab:related_works}
\end{table}
    \section{Preliminaries}

\subsection{Networks}
We now describe the network type within which a team of agents operates.

\paragraph*{Temporal graphs}
A temporal graph $\tmpG = (G, \varrho)$ is a graph where the presence of edges may vary over time. So,
\begin{itemize}
    \item $G=(V,E)$ is an undirected simple finite graph with $|V|=n$ nodes and $|E|=m$ edges;
    \item $\varrho:E\times \NatO \to \{0,1\}$ defines the presence (when 1) or absence (when 0) of edge $e$ at time $t$, where the time domain is $\NatO = 0,1,2\dots$.
\end{itemize}
We call the graph $G$ the \emph{underlying graph} of $\tmpG$.
Note that, if $\varrho=1$, then $\tmpG$ is a common \emph{static graph}.
Let $E_t = \{e\in E \sucht \varrho(e,t ) = 1\}$.
A temporal graph $\tmpG$ can be written also as the infinite sequence $\{G_t\}_{t\in \NatO}$ where each $G_t = (V, E_t)$ is the \emph{snapshot} of $\tmpG$ at time $t$ and contains only the present edges at time $t$.
An edge $e$ in $\tmpG$ is called \emph{recurrent} if, for any $t\in\NatO$, there exists a time $t'> t$ such that $\varrho(e,t') = 1$.
Otherwise, (i.e., if there exists a time $t\in\NatO$ such that $\varrho(e,t')=0$ for any $t'> t$), then $e$ is called \emph{transient}.

A \emph{journey} in $\tmpG$ is a sequence $J = ((e_1, t_1) , \dots, (e_k, t_k))$ such that $(e_1, \dots, e_k)$ defines a walk in $G$, and it holds that $t_i < t_{i+1}$ and $\varrho(e_i, t_i) = 1$ for any $i\in [1,k]$. In other words, a journey defines a temporal walk for an agent on the nodes of $G$.
We denote with $\J(v,w,t)$ the set of all the journeys from $v$ to $w$ starting at time $t'\geq t$.
A temporal graph is \emph{temporally connected} if $\J(v,w,t)\neq \varnothing$ for any pair of nodes $v,w$ and any time $t\in \NatO$.
Indeed, for $\tmpG$ to be temporally connected, the underlying graph $G$ must be connected; yet, the connectivity of $G$ is not sufficient to have temporal connectivity in $\tmpG$ due to the presence of transient edges.

\subparagraph*{Nodes, edges and (sub-)ports.}
Given a temporal graph $\tmpG = (G=(V,E), \varrho)$, we assume that the nodes are \emph{anonymous}, i.e., without IDs.
However, each node $v$ contains a local persistent memory called \emph{whiteboard}, i.e., a variable $\wboard_v$.
This variable has R/W access in mutual exclusion, and it is initially blank.

Given a node $v\in V$, we denote with $E(v)\subseteq E$ the set of edges in $G$ that are incident to $v$ in $G$, and we denote the \emph{degree} of $v$ as $\delta(v) = |E(v)|$.
We denote with $\Delta(\tmpG) := \Delta(G) = \max_{v\in V}\{\delta(v)\}$. 
When no ambiguity occurs, we will simply use $\Delta$.

For any incident edge, $v$ contains a \emph{port}, labeled by a bijection $\lambda_v: E(v) \to [0,\delta(v)-1]$.
This labeling is fixed (i.e., it does not change), local (i.e., it can be seen only in $v$), and arbitrary (i.e., it has no dependencies with other elements of $\tmpG$).
We indicate with $\port{v,w}$ the port at node $v$ from which the edge $\{v,w\}$ starts.
For convenience, we will use $\lambda_v(\port{v,w})$ in place of $\lambda_v(\{v,w\})$.
We can construct the global labeling function $\Lambda$ for $\tmpG$ such that $\Lambda(v,w) = \lambda_v(\port{v,w})$.
We refer to $\tmpG = (G=(V,E), \varrho, \Lambda)$ as a port-labeled temporal graph.

Generally, an edge $\{v,w\}\in E$ can be traveled in both senses at time $t\in \NatO$, if $\varrho(\{v,w\}, t) = 1$.
To distinguish the two senses, we say that each edge has two \emph{channels} that we denote with the ordered notation $(v,w)$ and $(w,v)$.
Thus, for a port $\port{v,w}$, there are two \emph{sub-ports}: the \emph{ingoing} one $\inport{v,w}$ (i.e., the endpoint of channel $(w,v)$) and the \emph{outgoing} one $\outport{v,w}$ (i.e., the endpoint of channel $(v,w)$).

Note that, if $\{v,w\}\in E$, then $\port{v,w}$ ($\port{w,v}$, resp.) and the related sub-ports are always present in $v$ ($w$, resp.) independently from the temporary presence of the edge $\{v,w\}$.

\paragraph*{Dangerous networks}
Given a port-labeled temporal graph $\tmpG = (G=(V,E), \varrho, \Lambda)$, we consider the dangerous version $\tmpG =(G=(V,E), \varrho, \Lambda, B)$ where $B\subset V$ represents the subset of nodes for which there exist no outgoing ports.
In other words, if an agent enters a node $b\in B$, it will stay there forever (we can assume the agent is destroyed as soon as it enters $b$).
For this reason, such nodes are called \emph{black holes}.

On the contrary, the nodes in $V\setminus B$ are called \emph{safe nodes}; they always contain both the ingoing and the outgoing sub-port for each incident edge.
Extending the terminology, we will say that an edge $\{v,w\}$ and the corresponding ports $\port{v,w}$ and $\port{w,v}$ are \emph{safe} (\emph{dangerous}, resp.) if both $v$ and $w$ are safe nodes (if $v$ or $w$ is a black hole, resp.)\footnote{For the sake of completeness, all the ports within a black hole are considered dangerous.}.

In the remainder, we will consider a team of agents operating on a \emph{network} defined as the port-labeled temporal dangerous graph $\tmpG =(G=(V,E),\varrho, \Lambda, B)$.
For convenience and w.l.o.g., we always assume that $E$ does not contain any edge between black holes (such edges, in fact, will never be traveled and are irrelevant for any task on dangerous networks).
Moreover, since the black holes may prevent agents from reaching some safe nodes of the graph, for exploring a temporal dangerous network $\tmpG=(G=(V,E),\varrho, \Lambda, B)$ it is necessary that for each pair of nodes $v,w\in V$ and any time $t\in \NatO$, $\J(v,w,t)$ contains at least one journey which does not contain black holes except possibly for the endpoints $v$ or $w$. 
In the following, we will make such an assumption\footnote{Note that a network with $|B|>1$ black holes can be seen as a network with a unique black hole to which all the dangerous edges are incident. However, in this case, we should assume that the underlying graph $G$ contains multiple edges (since a safe node $v$ may contain multiple dangerous ports). For consistency with the existing literature, we prefer assuming that $G$ is simple and thus that $|B|$ can be any.}.
\subparagraph*{Temporal models.}
We consider a hierarchy of temporal models to analyze the computational power of agents.
Given a temporally connected network $\tmpG$, we say that $\tmpG$ belongs to:
\begin{itemize}
	\item \TRAN{} (\emph{transient} model): when $\varrho$ allows transient edges;
	\item \RECU{} (\emph{recurrent} model): when $\varrho$ only allows recurrent edges;
	\item \STAT{} (\emph{static} model): when $\varrho = 1$ (i.e., the network is static).
\end{itemize}
By default, we consider the {\TRAN} model.

\subsection{Agents}	
    A \emph{team of agents} is a set $\agents=\{a_1, \dots, a_k\}$ of $k$ computational entities operating on a network $\tmpG =(G,\varrho, \Lambda, B)$.
	Agents are externally \emph{indistinguishable}.

    \subparagraph{Considered settings.}
    We will consider both the setting where agents are \emph{anonymous} and the setting where each agent has an internal \emph{unique ID}.
    Moreover, we consider both the setting where agents are \emph{oblivious} (i.e., devoid of any personal persistent memory), and the setting where each agent $a$ has a personal persistent memory, called \emph{notebook}, and denoted with $\nbook_a$.
    Generally, agents do not have any prior knowledge of the network (graph class, topology, number of black holes, recurrent edges, etc.) or of the team (size, initial positions, etc.).
    We can only assume that agents know the number of safe nodes $|V\setminus B|=s$.

    \subparagraph*{Positions and visibility.}
    At time 0, the agents are arranged at the center of some \emph{arbitrary} nodes of $G$ (multiplicities are allowed).
    During the evolution of the algorithm, each agent can be located in the center of a node (denoted with $\cnode$), on its ingoing or outgoing sub-ports, or traveling along the present edges.
    An agent $a$ has only \emph{local visibility}, i.e., it can only see the content of the node $v$ where $a$ is located at a given time. 
    In particular, $a$ can see the content of the whiteboard (in mutual exclusion), the labels of the ports, how many agents there are on $v$, and if these agents are on some sub-ports or on the center of the node.
    Yet, an agent cannot see if the edge of a port is present or not at a given time.

    \subparagraph*{Activation in mutual exclusion.}
    Time is divided into discrete time steps $t\in\NatO$ called \emph{rounds}.
    Agents are activated by \emph{semi-synchronous schedulers} (\ssynch): at any time $t\in\NatO$, an arbitrary subset of robots is activated according to the \emph{fairness condition} (each agent is activated infinitely often).
    We can formalize a {\ssynch} scheduler as a function 
    $\AS:\NatO \to 2^\agents  $ 
    such that $\forall a\in \agents, t\in \NatO$ there exists a time $t'>t$ such that $a\in \AS(t')$.
    Since agents operate on temporal graphs whose edge presence is decided by an adversary, it is necessary to assume another ``fairness'' condition on agents' activation: we assume the \emph{eventual transport} condition, which states that an agent located in an outgoing port of a recurrent edge will be eventually activated when the edge is present \cite{GotohFMS21,LunaDFS20}. 
    
    Within each round, the activated robots perform each one a \emph{Look-Compute-Move} (LCM) cycle. 
    Since in the Look and Compute steps agents read and write on the whiteboard, and the access to the whiteboard is in mutual exclusion, we can describe the sequential access to the whiteboard at round $t$ as a nested level of scheduling within each $\AS(t)$.
    So, we define  
    $\AS:\NatO \to \bigcup_{i=1}^{k}Perm(\agents, i)$
    where $Perm(S, i)$ represents the set of all the permutations obtained with $i$ elements taken from a set $S$.
    For example, if $\AS(3) = a_2a_9a_1a_0$, then at round $3$ the agents $a_2,a_9,a_1,a_0$ will be activated and will access the whiteboard in this order.

    \subparagraph*{Computation.}
    Let $\AS(t) = a_{i_1},\dots, a_{i_{h(t)}}$ be the sequence of activated agents at time $t$.
    All the agents in $\AS(t)$ execute the following steps (refer to \Cref{algo:LCM} in \Cref{appendix:pseudocodes}):
    
     \textbf{Look}: 
			Let $a$ be an agent in $\AS(t)$.
            Let $v$ be the node where agent $a$ is located.
            Then $a$ observes and takes these data: $pos(a)$, i.e., its position within $v$ (center or a sub-port), and the content of $\nbook_a$.
		
	\textbf{Compute}: 
            In this step, agents take access to the whiteboard in mutual exclusion.
            Thus, this step is executed in sequence, starting from $a_{i_1}$ to $a_{i_{h(t)}}$.
			Let $a$ be an agent among them during its turn.
            Then, $a$ grants access to $\wboard_v$, and reads its content.
            Let $\phi_v: \{\cnode\}\cup \{in,out\}\times[0,\delta_v-1]  \to [0,k]$ be the agents' arrangement\footnote{Indicating how many agents there are on the center $\cnode$ or on each sub-port of $v$.} on $v$ during the Compute step of $a$.
            Let $\sigma = \langle pos(a),\nbook_a ,\wboard_v,\phi_v \rangle$ be the tuple containing all the data got during the Look and Compute steps.
            Now, $a$ executes 
            $\algo(\sigma) = (p, \chi)$ where $\algo$ is the deterministic algorithm shared by all the agents, and where $p\in [0,\delta(v)-1]\cup \{\cnode, \NaN\}$ and $\chi\in \{0,1\}^*\cup \{\NaN\}$ is a binary word or undefined\footnote{We use the value $\NaN$ (not a number) to indicate that a function is not definited for that input.}.
			If $\chi$ is defined (i.e., $\chi\neq \NaN$), then $a$ writes $\chi$ on $\wboard_v$.
			Then, if $p$ is defined (i.e., $p\neq \NaN$), then $a$ moves to the outgoing port labeled as $p$ (if $p\in[0,\delta(v)-1]$) or at the center of $v$ (if $p=\cnode$); otherwise, $a$ stays still.

	\textbf{Move}: 
            Once all the agents in $\AS(t)$ have completed the Compute step, then each agent $a\in\AS(t)$ executes its movement in parallel with the other agents.
            In particular, if $a$ is at an outgoing port of some channel $(v,w)$, and $\varrho(\{v,w\}, t) = 1$, then $a$ travels along the channel $(v,w)$ and reaches the ingoing port of $w$.
			Otherwise, $a$ does nothing.

Through the infinite repetition of LCM cycles and execution of the algorithm $\algo$, the team aims to solve a common problem.

\subsection{Problems}

\subsubsection{Safe exploration}
	A node $v$ is \emph{visited} at time $t$ if at least one agent is located in $v$ (in its center or in one of its sub-ports).
    As soon as an agent $a$ visits a black hole, we say that $a$ gets \emph{trapped} in that node since it will no longer move to another node. 
    Conversely, $a$ is \emph{free} as long as it is not trapped.
    For a dangerous network $\tmpG = (G=(V,E),\Lambda, \varrho, B)$, we consider the following problems: 
	
	\begin{definition}[$\SEperp$]
		Given $k$ agents arranged on nodes of $\tmpG$, the \textsc{Safe Perpetual Exploration} ($\SEperp$) problem requires each safe node of $\tmpG$ to be visited infinitely often.
	\end{definition}
	
	\begin{definition}[$\SEfini$]
		Given $k$ agents arranged on nodes of $\tmpG$, the \textsc{Safe Finite Exploration} ($\SEfini$) problem requires the agents to reach a configuration in finite time where:
		\begin{itemize}		
			\item at least an agent is free;
            \item each safe node of $\tmpG$ has been previously visited at least once;
            \item each free agent is located at the center of a node, not necessarily the same;
			\item each agent will do nothing in the next time steps.
		\end{itemize}
	\end{definition}
	
	We are interested in the following further properties:
	\begin{itemize}
		\item \textbf{Loss-minimizing}: for both $\SEperp$ and $\SEfini$, we are interested in solutions that minimize the number of trapped agents.
		\item \textbf{Explicit termination}: for $\SEfini$, in addition to exploration completion (i.e., all nodes have been visited) and the termination of the algorithm (i.e., agents will no longer move), we require this information to be propagated to all safe nodes so that all agents get aware of the problem termination. We say that an agent \emph{gets aware} of the termination of the problem if it writes $\TERMINATION$ in its notebook\footnote{Oblivious agents can write $\TERMINATION$ on the whiteboard of the safe nodes to explicitly terminate.}.
        To avoid falling into trivial conditions, the awareness must be reached by the free agents only after solving $\SEfini$.
	\end{itemize}

\subsubsection{Port marking}
Safely exploring a dangerous network requires agents to probe and mark (in the best case, all) the ports as {\SAFE} or {\DANGEROUS}, thereby preventing all agents from being trapped in the black holes. 
We generally refer to this sub-routine as \emph{port marking} (\PM). As we will see, some ports may remain unmarked.

Specifically, a \emph{correct port marking} of a dangerous network $\tmpG = (G=(V,E),\Lambda, \varrho, B)$ is achieved when a team of agents reaches a configuration in a finite time $t\in \NatO$ where:
    \begin{itemize}
        \item at least one agent is free;
        \item all ports are \emph{correctly marked}. Namely, if $\port{v,w}$ is dangerous, it must be marked in $\wboard_v$ as {\DANGEROUS}; otherwise (i.e., $\port{v,w}$ is safe), it must be marked as {\SAFE};
        \item these values cannot be changed in the next time steps $t'\geq t$.
    \end{itemize}
Thus, a correct {\PM} presents neither false positives (i.e., safe ports marked as {\DANGEROUS}), nor false negatives (i.e., dangerous ports marked as {\SAFE}). 
Note that this does not preclude the possibility that, before $t$, some ports were incorrectly classified by the agents.
However, once the correct marking is reached at time $t$, it remains invariant thereafter and can be used by the free agents to safely explore the network.

As we will see in the next sections, it is not always possible to achieve a correct {\PM} in the {\TRAN} setting.
Therefore, we define the following three weaker variants.

A \emph{highly reliable port marking} of $\tmpG$ is achieved when a team of agents reaches a configuration in a finite time $t\in \NatO$ where: 
\begin{itemize}
        \item at least one agent is free;
        \item if $\{v,w\}$ is recurrent, then $\port{v,w}$ and $\port{w,v}$ are correctly marked as {\SAFE} or {\DANGEROUS};
        \item all the dangerous ports are correctly marked as {\DANGEROUS};
        \item these values cannot be changed in the next time steps $t'\geq t$.
    \end{itemize}
Thus, a highly reliable {\PM} differs from a correct one only in that it admits false positives only for ports belonging to transient edges; consequently, the set of the {\SAFE} ports is a subset of the actual safe ports in $\tmpG$.
Yet, this fact does not prevent the agents from reaching all the safe nodes of $\tmpG$.
In fact, let $\tmpGSAFE$ be the network obtained by removing from $\tmpG$ all the black nodes and all the edges with at least one port not marked as {\SAFE} in $\tmpG$.
Since $\tmpGSAFE$ contains all the safe recurrent edges of ${\tmpG}$, we can conclude that $\tmpGSAFE$ is temporally connected as $\tmpG$.

A \emph{reliable port marking} of $\tmpG$ allows false positives regardless of the nature of the edge, but with the constraint that $\tmpGSAFE$ must be temporally connected.
Formally, the achieved configuration must have:
\begin{itemize}
        \item at least one agent is free;
        \item all the dangerous ports are correctly marked as {\DANGEROUS};
        \item $\tmpGSAFE$ must be temporarily connected;
        \item these values cannot be changed in the next time steps.
\end{itemize}
Lastly, a \emph{weakly reliable port marking} of $\tmpG$ is reached when:
\begin{itemize}
        \item at least one agent is free;
        \item if a port is {\SAFE}-marked, it is truly safe;
        \item these values cannot be changed in the next time steps.
\end{itemize}
A correct {\PM} is highly reliable, which is, in turn, reliable, which is, in turn, weakly reliable.

    \subsection{Basic impossibilities}
We now define a property of networks which will be used to prove some impossibility results on the port marking problem (in \Cref{th:impossible_correct_PM} and later in \Cref{th:impossible_reliable_PM}).

\begin{definition}[Network indistinguishability]
	Given two networks $\tmpG = (G=(V,E), \Lambda, \varrho, B)$ and $\tmpG' = (G'=(V',E'), \Lambda', \varrho', B')$, we say that $\tmpG$ during the period $[t_a,t_b]$ is indistinguishable from $\tmpG'$ during the period $[t'_a,t'_b]$, and we denote this as $\tmpG[t_a,t_b] \indis \tmpG'[t'_a,t'_b]$, if $t_b-t_a=t'_b-t'_a$ and if there exists a bijection $h:V \to V'$ such that:
	\begin{itemize}
        \item $b\in B  \iff h(b) \in B'$;
        \item for any $v\in V\setminus B$, $\delta(v) = \delta(h(v))$;
		\item for any $t\in[0,t_b-t_a]$, it holds that $\{v,w\} \in E(G_{t_a+t})  \iff \{h(v),h(w)\} \in E(G'_{t'_a+t}) $;
        \item if $\{v,w\} \in E(G_t)$ for some $t\in[t_a,t_b]$, then $\lambda_v(\port{v,w}) = \lambda_{h(v)}(\port{h(v),h(w)})$.
	\end{itemize}
\end{definition}
Refer \Cref{fig:indistinguishable_scenarios} for an example.

\begin{figure}[th]
  \resizebox{\textwidth}{!}{
    \begin{subfigure}[t]{0.45\textwidth}
        \centering
        \scalebox{0.7}{
            \begin{tikzpicture}[on grid]  
                \def\l{2}
                \def\h{2.5}
                \node[safenode]  (v)  at (-1*\l,0)   {$v$};
                \node[safenode]   (u) at (0,0) {$u$};
                \node[safenode]   (w) at (1*\l,0) {$w$};
                \node[ggraph]   (H) at (0,-1*\h) {$H$};

                \node[blackhole,label=above:$b$] (BH) at (0,1) {};
    
                \path[dashed] (v)     edge node [above]           {}  (u);
                \path[-] (w)     edge node [above]           {}  (BH);
                \draw[snakearrow] (v) -- (H);
                \draw[snakearrow] (u) -- (H);
                \draw[snakearrow] (w) -- (H);

            \end{tikzpicture}
        }
        \caption{$\tmpG$, where $\{v,u\}$ is always absent.}
    \label{subfig:indistinguishable_scenarios_a}
    \end{subfigure}
    \hfill
    \begin{subfigure}[t]{0.45\textwidth}
        \centering
          \scalebox{0.7}{
            \begin{tikzpicture}[on grid]  
                \def\l{2}
                \def\h{2.5}
                \node[safenode]  (v)  at (-1*\l,0)   {$v$};
                \node[safenode]   (u) at (0,0) {$u$};
                \node[safenode]   (w) at (1*\l,0) {$w$};
                \node[ggraph]   (H) at (0,-1*\h) {$H$};

                \node[blackhole,label=above:$b$] (BH) at (0,1) {};

                \path[-] (w)     edge node [above]           {}  (BH);
                \path[dashed] (u)     edge node [above]           {}  (BH);
                \path[dashed] (v)     edge node [above]           {}  (BH);
                \draw[snakearrow] (v) -- (H);
                \draw[snakearrow] (u) -- (H);
                \draw[snakearrow] (w) -- (H);

            \end{tikzpicture}
        }
    \caption{$\tmpG'$, where $\{v,b\}$ and $\{u,b\}$ are absent during the period $[0,t]$.}
    \label{subfig:indistinguishable_scenarios_b}
    \end{subfigure}
    }
\caption{Two indistinguishable networks $\tmpG[0,t]\indis\tmpG'[0,t]$ with the same static subgraph $H$, but different edges among $v$, $u$, and the black hole. Solid (dashed, resp.) straight lines represent recurrent (transient, resp.) edges.}
\label{fig:indistinguishable_scenarios}
\end{figure}
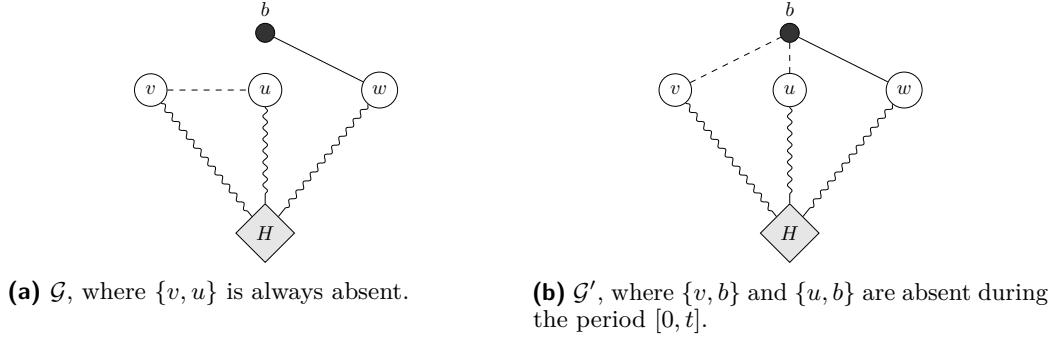  

A direct consequence of this property is that the same team of oblivious agents will have identical behavior when deployed on indistinguishable graphs if they start from indistinguishable configurations.
Formally, let $\agents$ be a group of $k$ oblivious agents deployed on $\tmpG$, $\algo$ be the deterministic algorithm they execute, $\AS$ be the scheduler activating $\agents$.
Let $\tmpG[t_a,t_b] \indis \tmpG'[t'_a,t'_b]$ for some $\tmpG'$ graph according to some bijection $h$. 
Now, suppose $\agents$ is arranged on $\tmpG'$ so that, for all $v\in V\setminus B$:
\begin{itemize}
    \item $\wboard_v$ at time $t_a$ is equal to $\wboard_{h(v)}$ at time $t'_a$;
    \item $\varphi_v(t_a) = \varphi_{h(v)}(t'_a)$, where $\varphi_v(t)$ defines the position of all the agents within $v$ at time $t$. 
\end{itemize}
Suppose that, while on $\tmpG'$, agents are activated by a scheduler $\AS'$ so that $\AS'(t_a'+t) = \AS(t_a+t)$ for $t\in [0,t_b-t_a]$.
Then, during $[t'_a,t'_b]$, $\agents$ performs on $\tmpG'$ the same actions as on $\tmpG$.

\begin{theorem}\label{th:impossible_correct_PM}
    It is impossible to achieve a correct port marking under the $\TRAN$ model even when agents know $|V|$, $|B|$, $k$, have infinite memory, and are activated under ${\fsynch}$.
\end{theorem}
\begin{proof}
    The claim derives from the impossibility of correctly marking transient edges.
	Let us proceed by contradiction, and let $\algo$ be an algorithm providing a correct port marking in the {\fsynch}-$\TRAN$ model.
	Let $\tmpG$ be the dangerous network depicted in \Cref{subfig:indistinguishable_scenarios_a} where the two safe nodes $v$ and $u$ are connected by a transient edge.
    Suppose the edge $\{v,u\}$ permanently disappears at time $0$.
    Since $\algo$ correctly marks the ports, let $t$ be the finite time when both $\port{v,u}$ and $\port{u,v}$ will be permanently marked as {\SAFE} in $\tmpG$.

    Now assume the same team is arranged on a graph $\tmpG'$ (depicted in \Cref{subfig:indistinguishable_scenarios_b}) which is identical to $\tmpG$, except that the edge $\{v,u\}$ is replaced with two transient edges $\{v,b\}$ and $\{u,b\}$.
	Suppose that the edges $\{v,b\}$ and $\{u,b\}$ are always absent.
    Then, since $\tmpG[0,\infty]\indis\tmpG'[0,\infty]$, we have that $\algo$ marks $\port{v,b}$ and $\port{u,b}$ in $\tmpG'$ as $\SAFE$ at time $t$, and such marking will remain invariant forever: this contradicts the fact that $\algo$ provides a correct port marking.
\end{proof}

    \section{Safe Perpetual Exploration}\label{sec:SE_perpetual}
\subsection{Assumptions and limits}
We study the $\SEperp$ problem for a team of anonymous and oblivious agents, under the $\ssynch$ setting, arbitrarily deployed on the safe nodes of a dangerous, anonymous, and $\TRAN$ network $\tmpG = (G=(V,E), \Lambda, \varrho, B)$ with $|V|=n$ nodes and $0 \leq |B| < n$ black holes.
Agents do not have any prior knowledge about the network or the team.
Under these assumptions, we first observe that it is impossible to explore (perpetually or finitely) the network with fewer than $\sum_{b\in B}\delta(b)$ agents.
In fact, since agents do not know the size of the network, any unexplored edge may lead to an unexplored leaf safe node; thus, each edge must be traversed at least once in order to explore all the nodes.
It follows that any algorithm solving $\SEperp$ under these assumptions will have at least $\sum_{b\in B}\delta(b)$ trapped agents, one for each dangerous port. 
This also defines the lower bound on the loss of any exploration solution.
Moreover:
\begin{theorem}[\cite{GotohFMS21}]\label{th:SEperp_Gotoh}
    There exists a network $\tmpG = (G=(V,E), \Lambda, \varrho, B)$ with $B=\varnothing$ and $\eta$ transient edges for which it is impossible to solve $\SEperp$ with fewer than $2\eta$ agents.
\end{theorem}

Informally, the former lower bound derives from the fact that $2$ agents are needed to “deal with” the two endpoints of each transient edge. 
However, since in our case $B$ may be non-empty and some transient edges may be dangerous, we can combine the two lower bounds (i.e., $\sum_{b\in B}\delta(b)$ to explore the dangerous ports and $2\eta$ to deal with the transient edges) by considering the parameter $\eta'$, i.e., the number of transient edges which are not dangerous.
Thus, we will always assume that any instance $\tmpG = (G=(V,E), \Lambda, \varrho, B)$ of the $\SEperp$ problem starts with at least $2\eta' + \sum_{b\in B}\delta(b) +1$ agents arranged on $V\setminus B$.

\subsection{Algorithm for $\SEperp$}
Under the above assumptions, we provide an algorithm, called $\Aperp$, which solves $\SEperp$ by combining the rotor-router mechanism with a cautious walk technique.
$\Aperp$ needs a $O(\Delta)$-size whiteboard\footnote{The precise upper bound is given by $O(\hat{\Delta})$ where $\hat{\Delta}$ corresponds to the maximal degree of the safe nodes in $\tmpG$. Yet, we use $\Delta$ for consistency with the literature.}.
By \Cref{th:impossible_correct_PM}, the final port marking provided by $\Aperp$ may not be correct; yet, it is always highly reliable.
\begin{center}
\footnotesize{
\renewcommand{\arraystretch}{1.4}{
    \begin{tabular}{|cccc|cc| c|c|}
        \hline
        \multicolumn{8}{|c|}{\textbf{Algorithm} $\Aperp$ -- \textbf{Safe Perpetual Exploration}, $\SEperp$}\\
        \hline
        \textbf{Network} & \textbf{Synch} & \textbf{Initial} & \textbf{Params} & \textbf{Mem} & \textbf{WhiteB} & \textbf{AgentId} & \textbf{Port marking}\\
        \hline
        \hline
        \rowcolor{yellow!30} {\TRAN} & {\ssynch} & Any & \crossmark & \crossmark & $\Delta$ & \crossmark & Highly reliable\\
        \hline
    \end{tabular}
}}
\end{center}

\subparagraph{Initialization.}
Each agent $a\in \agents$ is initially located at the center of an arbitrary safe node.
The whiteboard $\wboard_v$ of each safe node $v\in V\setminus B$ contains the following variables:
\begin{itemize}
    \item the pointer $\rotor_v \in [0,\delta(v)-1]$ which points to one port, 0 by default. This pointer will be used to implement the rotor-router mechanism;
    \item the array $\ports_v \in \{\SAFE, \UNEXPLORED, \DANGEROUS\}^{\delta(v)}$ which contains the state of each port. Namely, $\ports[i]$ indicates the state of port $i$.
    This array will be used to implement the port marking. At time 0, each port is marked as {\UNEXPLORED}.
\end{itemize}
Thus, the size of $\wboard_v$ is $\Theta(\delta_v)$ bits.

\subparagraph{Strategy.}
From a high-level perspective, our algorithm follows this scheme: unexplored ports are first explored cautiously, and, if safe and recurrent, they will eventually be marked as such; the rotor-router mechanism guarantees a complete exploration of all the nodes of the network, and eventually a perpetual exploration of all and only the safe nodes.

Let us now give the details of our algorithm.
The choice of the next edge to be traversed for an agent $a$ at a node $v$ is computed by the utility $\nextcandidate$, which locally computes the next candidate port for exploration within a node $v$ (refer to \Cref{algo:next_candidate} for the pseudocode). 
A \emph{candidate} port must be \emph{empty} (i.e., without any agent on its sub-ports) and not marked as {\DANGEROUS}.
Thus, $\nextcandidate{}$ returns the first empty port marked as {\SAFE} or {\UNEXPLORED} within $v$, starting from the port pointed by $\rotor_v$, or $\NaN$ if no candidate exists.
In the first case, $a$ positions itself on the outgoing sub-port returned by $\nextcandidate{}$ and increments the value of $\rotor_v$ by one. Otherwise, it does nothing and waits for a candidate.

Before the exploration of an {\UNEXPLORED} port $\port{v,w}$, the agent $a$ marks it as {\DANGEROUS} in $\ports_v$; if the agent will eventually reach $\inport{w,v}$, it marks such a port as {\SAFE} in $\ports_w$ if it was marked otherwise.
Then, the agent could be required to go back along the channel $(w,v)$ to mark $\port{v,w}$ as {\SAFE} as well; this happens only if $\port{w,v}$ was marked as {\UNEXPLORED} before the arrival of $a$.
Let us now describe in detail the algorithm.

\subparagraph{Algorithm scheme.}
Let $\AS(t)$ be the set of agents activated at time $t$. 
We now describe $\Aperp$ by defining the actions made sequentially by each agent $a\in \AS(t)$ during its Compute step.
Refer to \Cref{algo:SE_perpetual} for the pseudocode of $\algo^\infty$.
Let $v$ be the node where $a$ lies at the beginning of the round $t$: if $v$ is a black hole, $a$ does nothing (w.l.o.g., we can assume that $a$ is destroyed once in a black hole).
Otherwise, the actions of $a$ depends on $pos(a)$ as follows, where $pos(a)$ is the relative position of $a$ within $v$:
\begin{itemize}
    \item if $pos(a) = \cnode$ (i.e., the center), then $a$ reads from $\wboard_v$ the value of $\rotor_v$, and computes the value $j=\nextcandidate{}$.
	Then, $a$ behaves as follows:
	\begin{itemize}
        \item if $j = \NaN$ (i.e., there are no candidate ports in $v$), then $a$ does nothing;
		\item otherwise (i.e., $j\in [0,\delta_v-1]$), $a$ acts as follows:
        \begin{itemize}
            \item $a$ moves to the outgoing sub-port $\lambda^{-1}(j)$;
            \item $a$ updates the value of $\rotor_v$ so that now it points to $j+1 \mod \delta_v$;
            \item if $\ports_v[j] = \UNEXPLORED$, then $a$ locks the port $\lambda^{-1}_v(j)$ by updating the corresponding state $\ports_v[j] \gets \DANGEROUS$.
        \end{itemize}
		
	\end{itemize}
	
	\item if $pos(a) = \inport{v,w}$ (i.e., $a$ is activated after its Move along the channel $(w,v)$), then:
	\begin{itemize}
        \item if $\ports_v[\lambda_v(w)] = \UNEXPLORED$, then $a$ updates the state of the port to $\SAFE$ and moves to the outgoing port $\outport{v,w}$ in order to go back to the node $w$;
        \item if $\ports_v[\lambda_v(w)] =  \DANGEROUS$, then $a$ updates the state of the port to $\SAFE$ and moves to the center of $v$. In this case, $a$ does not need to go back to $w$ since there is already an agent (the one which has marked the port as {\DANGEROUS}) that is moving/has moved along the channel $(v,w)$ and it will mark the port $\port{w,v}$ as {\SAFE}.
		\item if $\ports_v[\lambda_v(w)] = \SAFE$, then $a$ moves to the center of $v$. Even in this case, $a$ does not need to go back to $w$ since there is already an agent (the one that has marked the port as {\SAFE}) that is moving/has moved along the channel $(v,w)$, and that will mark the port $\port{w,v}$ as {\SAFE} in case.
	\end{itemize}
	
	\item if $pos(a) = \outport{v,w}$, it does nothing.

\end{itemize}

During the Move step, each agent in $\AS(t)$ that is located in an outgoing sub-port of a safe node will travel along the corresponding edge if such an edge is present at time $t$.
In particular, if $a$ is located at $\outport{v,w}$ during the Move step of round $t$, and if $\varrho(\{v,w\}, t) = 1$, then $a$ will be located at $\inport{w,v}$ at the beginning of round $t+1$.

\subsection{Analysis of $\Aperp$}
We here prove the properties and the correctness of $\Aperp$ in solving $\SEperp$ (proofs in \Cref{sec:proofs_Aperp}).

\begin{lemma}\label{lemma:SEperp_safe}
    If a port $\port{v,w}$ is marked as {\SAFE} at round $t$, then $w$ is truly a safe node and $\port{v,w}$ permanently remains marked as {\SAFE} thereafter.
\end{lemma}

\begin{lemma}\label{lemma:SEperp_dangerous}
    If a port $\port{v,w}$ is marked as {\DANGEROUS} at round $t$ and node $w$ is truly a black hole, then such a port will be permanently marked as {\DANGEROUS}.
\end{lemma}

\begin{lemma}\label{lemma:SEperp_false_dangerous}
    If a port $\port{v,w}$ is marked as {\DANGEROUS} at round $t$ and node $w$ is safe and $\{v,w\}$ is recurrent, then $\port{v,w}$ will be eventually and permanently marked as {\SAFE}.
\end{lemma}

\begin{lemma}\label{lemma:always_agent_travels}
    Consider a team of at least $2\eta' + \sum_{b\in B}\delta(b) +1$ agents executing $\Aperp$ on a network $\tmpG = (G=(V,E), \Lambda, \varrho, B)$. Then, for each round $t$, there exists a time $t'\geq t$ where at least one agent travels along an edge.
\end{lemma}

\begin{lemma}\label{lemma:recurrent_edges_visited_inf}
    If a safe node $v$ is visited by an incoming agent an infinite number of times, then each safe node $w$ connected to $v$ through a recurrent edge will be visited an infinite number of times.
\end{lemma}

\begin{theorem}\label{th:SE_perpetual}
	$\algo^\infty$ ensures a team of at least $2\eta' + \sum_{b\in B}\delta(b) +1$ oblivious and anonymous agents arranged on the anonymous nodes of a temporally connected network 
	$\tmpG$ under the {\ssynch}-{\TRAN} model to solve $\SEperp$ without any prior knowledge of $\tmpG$.
    Eventually, $\algo^\infty$ provides a highly reliable port marking on $\tmpG$.
\end{theorem}

    \section{Safe Finite Exploration}\label{sec:SE_finite}
\subsection{Assumptions and limits}
We now extend the previous algorithm so that agents can detect when every node has been visited and terminate.
To prevent agents from perpetually exploring the same indistinguishable nodes, we now assume that \emph{(i)} nodes have IDs, or \emph{(ii)} there is a way to assign IDs to the nodes.
To remain consistent with the existing literature, we choose the second approach: we assume each agent has its own ID, say $a_i$ with $i\in [0,k-1]$.
As soon as $a_i$ visits a not-yet-named node $v$, it will assign a distinct ID to $v$ (by composing its agent ID and a progressive value), and it will write this ID on $\wboard_v$.
Note that assumptions \emph{(i)} and \emph{(ii)} are equivalent: even assuming \emph{(i)}, each agent can use the ID of the node where it is initially arranged as its own ID.
In case of multiplicity, the mutually exclusive access to the whiteboards can provide a simple mechanism for assigning IDs to co-located agents \cite{DobrevFPS06}.

Moreover, we assume that agents have a notebook and know $|V\setminus B|= s$; no other information is needed.
The notebook of each agent is used to maintain and disseminate the list of the already-visited nodes. Such lists are copied onto the whiteboards or merged with the existing ones.
As soon as a list contains $s$ IDs, the agent knows that all the safe nodes have been visited.
We now observe the main computational difference between $\SEfini$ and $\SEperp$.
\begin{theorem}\label{th:impossible_reliable_PM}
	Any solution for $\SEfini$ under the $\TRAN$ model cannot achieve a reliable port marking, even assuming agents/nodes with IDs and infinite memory, $\fsynch$ schedulers, and the knowledge of $k$, $|V|$, and $|B|$.
\end{theorem}
\begin{proof}
    The impossibility derives from the fact that agents cannot distinguish between a transient and a recurrent edge in finite time.
    Suppose by contradiction that there exists a solution $\algo$ for $\SEfini$.
    Assume $\algo$ is run on the network in \Cref{subfig:indistinguishable_recurrent_transient_a} where, from time 0, the recurrent edge $\{v,u\}$ is absent for an unpredictable but finite amount of time, while the transient edge $\{v,h\}$ is always present in that period.
    Since this network is indistinguishable from the network in \Cref{subfig:indistinguishable_recurrent_transient_b} where the dangerous edges $\{v,b\}$ and $\{u,b\}$ are absent for an unpredictable amount of time, eventually $\algo$ must terminate and so permanently mark the ports $\port{v,h}$ and $\port{h,v}$ as $\SAFE$ and the ports $\port{v,u}$ and $\port{u,v}$ as $\DANGEROUS$ in a finite $t$.
    However, the sub-network $\tmpGSAFE$ obtained by keeping only the ports marked as {\SAFE} has $\{\{v,h\}, \{u,h\}, \{w,h\}\}$ as edge-set; indeed, $\tmpGSAFE$ is not temporally connected (in fact, the edge $\{v,h\}$ will permanently disappear at time $t+1$, so that the safe node $v$ will remain indefinitely isolated).
    This contradicts the fact that $\algo$ provides a reliable port marking.
\end{proof}
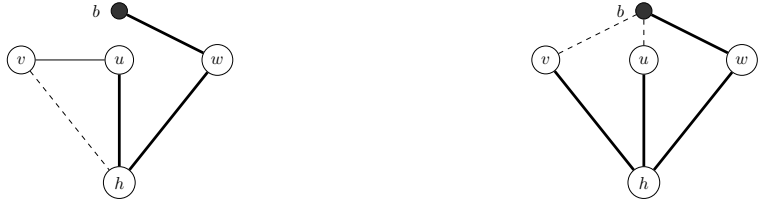
\begin{figure}[bth]
  \resizebox{\textwidth}{!}{
    \captionsetup{font=normalsize,labelfont={bf,sf}}
    \captionsetup[sub]{font=footnotesize,labelfont={bf,sf}}
     \captionsetup[subfigure]{width=0.9\textwidth}

    \begin{subfigure}[t]{0.45\textwidth}
        \centering
          \scalebox{0.6}{
            \begin{tikzpicture}[on grid]  
                \def\l{2}
                \def\h{2.5}
                \node[safenode]  (v)  at (-1*\l,0)   {$v$};
                \node[safenode]   (u) at (0,0) {$u$};
                \node[safenode]   (w) at (1*\l,0) {$w$};
                \node[safenode]   (h) at (0,-1*\h) {$h$};

                \node[blackhole,label={[label distance=1mm]left:$b$}] (BH) at (0,1) {};
                \path[staticedge] (w)     edge node [above]           {}  (BH);
               
                \draw[dashed] (v) -- (h);
                \draw[staticedge] (u) -- (h);
                \draw[staticedge] (w) -- (h);
                \draw[-] (v)-- (u);

            \end{tikzpicture}
            
        }
    \caption{$\{v,h\}$ always appears in $[0,t]$ and it permanently disappears at time $t+1$; $\{v,u\}$ never appears in $[0,t]$.}
    \label{subfig:indistinguishable_recurrent_transient_a}
    \end{subfigure}
    \hfill
    
    \begin{subfigure}[t]{0.45\textwidth}
        \centering
          \scalebox{0.6}{
            \begin{tikzpicture}[on grid]  
                \def\l{2}
                \def\h{2.5}
                \node[safenode]  (v)  at (-1*\l,0)   {$v$};
                \node[safenode]   (u) at (0,0) {$u$};
                \node[safenode]   (w) at (1*\l,0) {$w$};
                \node[safenode]   (h) at (0,-1*\h) {$h$};

                \node[blackhole,label={[label distance=1mm]left:$b$}] (BH) at (0,1) {};

                \path[staticedge] (w)     edge node [above]           {}  (BH);
                \path[dashed] (u)     edge node [above]           {}  (BH);
                \path[dashed] (v)     edge node [above]           {}  (BH);
                \draw[staticedge] (v) -- (h);
                \draw[staticedge] (u) -- (h);
                \draw[staticedge] (w) -- (h);

            \end{tikzpicture}
            }
    \caption{$\{v,b\}$ and $\{u,b\}$ are absent for an unpredictable time.}
    \label{subfig:indistinguishable_recurrent_transient_b}
    \end{subfigure}
        }
\caption{Solid (dashed, resp.) straight lines represent recurrent (transient, resp.) edges. Thick solid lines represent static edges (i.e., always present).}
\label{fig:indistinguishable_recurrent_transient}
\end{figure}

However, we will show that it is possible to provide a port marking that is weakly reliable.

\subsection{Algorithm for $\SEfini$}
Under the above assumptions, we now describe our algorithm, called $\Afini$, for solving $\SEfini$.
As we will see, $\Afini$ uses $\Aperp$ as a subroutine.
\begin{center}
\footnotesize{
\renewcommand{\arraystretch}{1.4}{
    \begin{tabular}{|cccc|cc|c|c|}
        \hline
        \multicolumn{8}{|c|}{\textbf{Algorithm} $\Afini$ -- \textbf{Safe Finite Exploration}, $\SEfini$}\\
        \hline
        \textbf{Network} & \textbf{Synch} & \textbf{Initial} & \textbf{Params} & \textbf{Mem} & \textbf{WhiteB} & \textbf{AgentId} & \textbf{Port marking}\\
        \hline
        \hline
        \rowcolor{yellow!30}{\TRAN} & {\ssynch} & Any & $s$ & $n\log{n}$ & $n\log{n}$ &\greencheck& Weakly reliable\\
        \hline
    \end{tabular}
}}
\end{center}
\subparagraph{Initialization:} 
As for $\Aperp$, we assume that the whiteboard of each safe node $v$ contains the variables $\ports_v$ and $\rotor_v$.
Moreover, each $\wboard_v$ contains three variables: $\nodestatus_v$, $\nodeid_v$ and $\nodelist_v$.
In particular, $\nodestatus_v$ will be used to contain the value $\EXPLCOMPLETE$ only when the exploration of all the safe nodes has been completed; $\nodeid_v$  will contain the ID of the node, which is given by the first agent accessing $\wboard_v$; $\nodelist_v$ will contain a cumulative list of the already-visited node IDs. 

Each agent executes the algorithm $(\Afini, s, a)$ where $\Afini$ is the same deterministic algorithm parametrized w.r.t. $s =|V\setminus B|$ and the personal agent ID $a$. 
Each agent maintains in its notebook a variable $\agentstatus_a\in\{\EXPLORATION, \PROPAGATION, \TERMINATION\}$---initialized to $\EXPLORATION$---which will be used to determine the action to be performed, a variable $\agentnodelist_a$ containing the list---initially empty---of the node IDs encountered during the exploration, and an integer variable $\counter_a$---initialized to 0---which will be used to provide IDs to the nodes.
The items of the lists $\nodelist_v$ and $\agentnodelist_a$ will be in the form $(\underline{x},\texttt{bool})$ where $\underline{x}=(a,j)$ defines a node ID, while \texttt{bool} is a boolean ({\false} by default) which will be set to {\true  } only during the propagation phase.
\Cref{table:variables_SEfini} summarizes the variables used in $\Afini$.

\subparagraph{Algorithm scheme:} 
The algorithm is composed of three subsequent phases: \emph{exploration}, \emph{propagation}, and \emph{termination}.
During the exploration phase, the agents continue exploring the nodes of the network essentially by executing the same strategy as in $\Aperp$. 
In addition, the agents provide nodes with IDs: such IDs are collected in the $\agentnodelist$ variables and distributed along the network in the $\nodelist$ variables.
This phase ends as soon as an agent has collected all the $s$ node IDs.

At this point, the agents that have detected that all safe nodes have been visited at least once (since these agents have already collected all the $s$ node IDs) must propagate this information to all other nodes.
To keep track of the nodes on which the information has already been propagated, the agents set the corresponding \texttt{bool} value to true in the lists $\nodelist$ and $\agentnodelist$.
Eventually, all the other $\nodelist$ variables will contain all the $s$ nodes in the form $(\underline{x},\texttt{true})$: then, the agents will be aware of the algorithm termination.

More formally, let $a$ be an agent activated in a safe node $v$. 
We now describe in detail the actions of $a$ when executing $(\algo, m, a)$ in one LCM cycle. 
Refer to \Cref{algo:SE_finite} in \Cref{appendix:pseudocodes}.
\begin{itemize}
    \item \textbf{Exploration.} (case $\agentstatus_a = \EXPLORATION$)
        \begin{itemize}
            \item if $\nodestatus_v \neq \EXPLCOMPLETE$, then
                \begin{itemize}
                	\item if $\nodeid_v$ is undefined, then $a$ gives a name to $v$ by updating the variable $\nodeid_v$.
                		In particular, $a$ writes $\nodeid_v \gets (a, j)$ on $\wboard_v$ where $j=0,1,2,\dots$ is the value of $\counter_a$.
                		After that, $a$ increases the value of the counter by one, and adds to $\agentnodelist_a$ the value $(a,j, \texttt{false})$.
                	\item $a$ computes $\ell = \nodelist_v \cup \agentnodelist_a$ (i.e., the comprehensive list of all the distinct node IDs appearing in the two variables) and updates the two lists: $\nodelist_v \gets \ell$ and $\agentnodelist_a\gets \ell$.
                	\item if $\ell$ contains $s$ node IDs, then $a$ updates $\agentstatus_a\gets \PROPAGATION$; at this point, the propagation phase starts. Otherwise (i.e., $\ell$ contains fewer than $s$ node IDs), the algorithm executes $\Aperp$ to continue the exploration.
                \end{itemize}
            \item otherwise (i.e., if $\nodestatus_v = \EXPLCOMPLETE$), $a$ sets $\agentstatus_a\gets \PROPAGATION$.
        \end{itemize}
    \item \textbf{Propagation:} (case $\agentstatus_a = \PROPAGATION$)
        \begin{itemize}
        	\item $a$ writes $\nodestatus_v\gets \EXPLCOMPLETE$ and sets $\texttt{bool}\gets \texttt{true}$ for the entry of node $v$ in $\agentnodelist_a$. Then, $a$ computes $\ell$ as the merge of $\agentnodelist_a$ and $\nodelist_v$. If the same node ID appears in both lists with two different \texttt{bool} values, that node ID is copied in $\ell$ only with the \texttt{true} entry;
            \item $a$ writes $\agentnodelist_a\gets \ell$ and $\nodelist_v\gets \ell$;
            \item if all the $s$ entries of $\agentnodelist_a$ have \texttt{bool}=\texttt{true}, then $a$ sets $\agentstatus_a\gets \TERMINATION$; otherwise $a$ executes $\Aperp$ to continue the propagation.
        \end{itemize}
    \item \textbf{Termination:} (case $\agentstatus_a = \TERMINATION$) Agent $a$ is aware that all the nodes have been visited at least once and are marked as $\EXPLCOMPLETE$. So, $a$ goes to the center of $v$ and terminates. Eventually, all free agents will set $\agentstatus\gets\TERMINATION$ and will terminate at the center of some safe nodes. The problem is solved.
\end{itemize}

\subsection{Analysis of $\Afini$}
We prove the properties and the correctness of $\Afini$ in solving $\SEfini$ (proofs in \Cref{sec:proofs_Afini}).

\begin{lemma}\label{lemma:SEfini_Exploration_finite}
    Consider a team of at least $2\eta' + \sum_{b\in B}\delta(b) +1$ agents executing $\Afini$ on a network $\tmpG = (G=(V,E), \Lambda, \varrho, B)$. Then, the Exploration phase is finite and ensures that all safe nodes are visited at least once and that at least one agent remains free.
\end{lemma}

\begin{lemma}\label{lemma:SEfini_Propagation_finite}
    The Propagation phase is finite, ensures that all safe nodes have $\nodestatus$ marked as {\EXPLCOMPLETE} and that at least one agent remains free.
\end{lemma}

\begin{lemma}\label{lemma:SEfini_Termination_finite}
    The Termination phase is finite and ensures that all free agents, of which are at least one, are at the center of a safe node.
\end{lemma}

\begin{theorem}\label{th:SE_finite}
	$\Afini$ ensures a team of at least $2\eta' + \sum_{b\in B}\delta(b) +1$ agents with IDs arranged on the anonymous nodes of a temporally connected network 
	$\tmpG$ under the {\ssynch}-{\TRAN} model to solve $\SEfini$ with explicit termination, only knowing $s=|V\setminus B|$.
    Eventually, $\Afini$ provides a \emph{weakly reliable} port marking on $\tmpG$.
\end{theorem}

    \section{Conclusions}
In this paper, we have extended the existing study on exploration of \emph{arbitrary} and \emph{unknown} dynamic networks by considering the possibility that the network is dangerous, namely, it may contain \emph{black holes}.
In particular, we have studied the capability of a team of \emph{semi-synchronous} agents to solve the \probfont{Safe Perpetual Exploration} ($\SEperp$) and the \probfont{Safe Finite Exploration} ($\SEfini$) problems, under the \emph{temporal connectivity} assumption, i.e., the minimal connectivity condition such that a non-trivial problem in a temporal graph can be solved.
Interestingly, $\SEperp$ can be solved under a very weak setting (anonymous nodes and agents, oblivious agents starting from any initial configuration, no prior knowledge of the team or the network).
For $\SEfini$, we propose an algorithm that makes agents explicitly terminate the exploration by assuming agents have IDs and a limited personal memory, only knowing the number of safe nodes.
A first future work may investigate whether $\SEfini$ can be solved under weaker assumptions.
Another interesting research direction is to consider more adversarial settings of dynamic networks; e.g., settings in which some edges may dynamically change one of their endpoint ports.

    \bibliography{refs}

@incollection{Das19_BOOK,
  author       = {Shantanu Das},
  title        = {Graph Explorations with Mobile Agents},
  booktitle    = {Distributed Computing by Mobile Entities, Current Research in Moving and Computing},
  series       = {Lecture Notes in Computer Science},
  pages        = {403--422},
  publisher    = {Springer},
  year         = {2019},
  doi          = {10.1007/978-3-030-11072-7_16}
}

@article{DevismesDL26,
  author       = {St{\'{e}}phane Devismes and
                  Yoann Dieudonn{\'{e}} and
                  Arnaud Labourel},
  title        = {Graph Exploration: The Impact of a Distance Constraint},
  journal      = {Algorithmica},
  volume       = {88},
  number       = {2},
  pages        = {24},
  year         = {2026},
  doi          = {10.1007/S00453-026-01378-4}
}

@inproceedings{BockenhauerFUW23,
  author       = {Hans{-}Joachim B{\"{o}}ckenhauer and
                  Fabian Frei and
                  Walter Unger and
                  David Wehner},
  title        = {Zero-Memory Graph Exploration with Unknown Inports},
  booktitle    = {Procs. of the 30th International Colloquium on Structural Information and Communication Complexity (SIROCCO)},
  pages        = {246--261},
  publisher    = {Springer},
  year         = {2023},
  doi          = {10.1007/978-3-031-32733-9_11}
}

@article{PattanayakP24,
  author       = {Debasish Pattanayak and Andrzej Pelc},
  title        = {Graph exploration by a deterministic memoryless automaton with pebbles},
  journal      = {Discrete Applied Mathematics},
  volume       = {356},
  pages        = {149--160},
  year         = {2024},
  doi          = {10.1016/J.DAM.2024.05.024}
}

@article{GotohFMS21,
  author       = {Tsuyoshi Gotoh and
                  Paola Flocchini and
                  Toshimitsu Masuzawa and
                  Nicola Santoro},
  title        = {Exploration of dynamic networks: Tight bounds on the number of agents},
  journal      = {Journal of Computer and System Sciences},
  volume       = {122},
  pages        = {1--18},
  year         = {2021},
  doi          = {10.1016/J.JCSS.2021.04.003}
}

@inproceedings{DereniowskiKPU14,
  author       = {Dariusz Dereniowski and Adrian Kosowski and Dominik Pajak and Przemyslaw Uznansk},
  title        = {Bounds on the Cover Time of Parallel Rotor Walks},
  booktitle    = {Procs. of the  31st International Symposium on Theoretical Aspects of Computer Science (STACS) },
  pages        = {263--275},
  publisher    = {LIPIcs},
  year         = {2014},
  doi          = {10.4230/LIPIcs.STACS.2014.263}
}

@article{KosowskiP19,
author       = {Adrian Kosowski and Dominik Pajak},
  title        = {Does adding more agents make a difference? {A} case study of cover time
for the rotor-router},
  journal      = {Journal of Computer and System Sciences},
  volume       = {106},
  pages        = {80--93},
  year         = {2019},
  doi          = {10.1016/J.JCSS.2019.07.001}
}

@misc{BileskiM26_arxiv,
      title={Don't Be Afraid to Die: Black Hole Search in Dynamic Graphs with Fewer Agents}, 
      author={Kass Bileski and Avery Miller},
      year={2026},
      eprint={2609.06850},
      archivePrefix={arXiv},
      primaryClass={cs.DC},
      url={https://arxiv.org/abs/2609.06850}, 
}

@inproceedings{KaurS26,
  author       = {Tanvir Kaur and
                  Ashish Saxena},
  title        = {When Agents are Powerful: Black Hole Search in Time-Varying Graphs},
  booktitle    = {Procs. of the 22nd International
                  Conference on Distributed Computing and Intelligent Technology ({ICDCIT})},
  volume       = {16420},
  pages        = {3--18},
  publisher    = {Springer},
  year         = {2026},
  doi          = {10.1007/978-3-032-16632-6_1}
}

@inproceedings{ShimoyamaSKM22,
  author       = {Kohei Shimoyama and
                  Yuichi Sudo and
                  Hirotsugu Kakugawa and
                  Toshimitsu Masuzawa},
  title        = {Invited Paper: One Bit Agent Memory is Enough for Snap-Stabilizing
                  Perpetual Exploration of Cactus Graphs with Distinguishable Cycles},
  booktitle    = {Procs. of the 24th International Symposium on Stabilization, Safety, and Security of Distributed Systems (SSS)},
  pages        = {19--34},
  publisher    = {Springer},
  year         = {2022},
  doi          = {10.1007/978-3-031-21017-4_2}
}

@article{MandalMM23,
  author       = {Subhrangsu Mandal and
                  Anisur Rahaman Molla and
                  William K. Moses Jr.},
  title        = {Efficient live exploration of a dynamic ring with mobile robots},
  journal      = {Theoretical Computer Science},
  volume       = {980},
  pages        = {114201},
  year         = {2023},
  doi          = {10.1016/J.TCS.2023.114201}
}

@article{LunaDFS20,
  author       = {Giuseppe Antonio Di Luna and
                  Stefan Dobrev and
                  Paola Flocchini and
                  Nicola Santoro},
  title        = {Distributed exploration of dynamic rings},
  journal      = {Distributed Computing},
  volume       = {33},
  number       = {1},
  pages        = {41--67},
  year         = {2020},
  doi          = {10.1007/S00446-018-0339-1}
}

@inproceedings{GotohSOKM18,
  author       = {Tsuyoshi Gotoh and
                  Yuichi Sudo and
                  Fukuhito Ooshita and
                  Hirotsugu Kakugawa and
                  Toshimitsu Masuzawa},
  title        = {Group Exploration of Dynamic Tori},
  booktitle    = {Procs. of the 38th {IEEE} International Conference on Distributed Computing Systems (ICDCS)},
  pages        = {775--785},
  publisher    = {{IEEE} Computer Society},
  year         = {2018},
  doi          = {10.1109/ICDCS.2018.00080}
}

@inproceedings{ErlebachKLSS19,
  author       = {Thomas Erlebach and
                  Frank Kammer and
                  Kelin Luo and
                  Andrej Sajenko and
                  Jakob T. Spooner},
  title        = {Two Moves per Time Step Make a Difference},
  booktitle    = {Procs. of the 46th International Colloquium on Automata, Languages, and Programming (ICALP)},
  pages        = {141:1--141:14},
  year         = {2019},
  doi          = {10.4230/LIPICS.ICALP.2019.141}
}

@inproceedings{ErlebachS18,
  author       = {Thomas Erlebach and Jakob T. Spooner},
  title        = {Faster Exploration of Degree-Bounded Temporal Graphs},
  booktitle    = {Procs. of the 43rd International Symposium on Mathematical Foundations of Computer Science (MFCS)},
  year         = {2018},
  doi          = {10.4230/LIPICS.MFCS.2018.36}
}

@article{IlcinkasW18,
  author       = {David Ilcinkas and
                  Ahmed Mouhamadou Wade},
  title        = {Exploration of the T-Interval-Connected Dynamic Graphs: the Case of
                  the Ring},
  journal      = {Theory of Computing Systems},
  volume       = {62},
  number       = {5},
  pages        = {1144--1160},
  year         = {2018},
  doi          = {10.1007/S00224-017-9796-3}
}

@article{MichailS16,
  author       = {Othon Michail and
                  Paul G. Spirakis},
  title        = {Traveling salesman problems in temporal graphs},
  journal      = {Theoretical Computer Science},
  volume       = {634},
  pages        = {1--23},
  year         = {2016},
  doi          = {10.1016/J.TCS.2016.04.006}
}

@inproceedings{Erlebach0K15,
  author       = {Thomas Erlebach and
                  Michael Hoffmann and
                  Frank Kammer},
  title        = {On Temporal Graph Exploration},
  booktitle    = {Procs. of the 42nd International Colloquium on Automata, Languages, and Programming (ICALP)},
  pages        = {444--455},
  publisher    = {Springer},
  year         = {2015},
  doi          = {10.1007/978-3-662-47672-7_36}
}

@inproceedings{KuhnLO10,
  author       = {Fabian Kuhn and
                  Nancy A. Lynch and
                  Rotem Oshman},
  title        = {Distributed computation in dynamic networks},
  booktitle    = {Procs. of the 42nd {ACM} Symposium on Theory of Computing (STOC)},
  pages        = {513--522},
  publisher    = {{ACM}},
  year         = {2010},
  doi          = {10.1145/1806689.1806760}
}

@article{DasFKNS07,
  author       = {Shantanu Das and
                  Paola Flocchini and
                  Shay Kutten and
                  Amiya Nayak and
                  Nicola Santoro},
  title        = {Map construction of unknown graphs by multiple agents},
  journal      = {Theoretical Computer Science},
  volume       = {385},
  number       = {1-3},
  pages        = {34--48},
  year         = {2007},
  doi          = {10.1016/J.TCS.2007.05.011}
}

@incollection{MarkouS19,
  author       = {Euripides Markou and
                  Wei Shi},
  title        = {Dangerous Graphs},
  booktitle    = {Distributed Computing by Mobile Entities, Current Research in Moving
                  and Computing},
  pages        = {455--515},
  publisher    = {Springer},
  year         = {2019},
  doi          = {10.1007/978-3-030-11072-7_18}
}

@misc{SaRMMoSh26,
  author       = {Ashish Saxena and Anisur Rahaman Molla and Kaushik Mondal and Gokarna Sharma},
  title        = {Semi-Synchronous Exploration in Dynamic Graphs},
  year         = {2026},
  eprint       = {arXiv:2605.14375},
  archivePrefix = {arXiv},
  url          = {https://arxiv.org/abs/2605.14375}
}

@inproceedings{KaurSMM25,
  author       = {Tanvir Kaur and
                  Ashish Saxena and
                  Partha Sarathi Mandal and
                  Kaushik Mondal},
  title        = {Black Hole Search in Dynamic Graphs},
  booktitle    = {Procs. of the 26th International Conference on Distributed Computing and Networking (ICDCN)},
  pages        = {221--230},
  publisher    = {{ACM}},
  year         = {2025},
  doi          = {10.1145/3700838.3700869}
}

@inproceedings{KaurSMM25_SSS,
  author       = {Tanvir Kaur and
                  Ashish Saxena and
                  Partha Sarathi Mandal and
                  Kaushik Mondal},
  title        = {Black Hole Search by Scattered Agents on Time-Varying Dynamic Graphs},
  booktitle    = {Procs. of the 27th International Symposium on Stabilization, Safety, and Security of Distributed Systems (SSS)},
  pages        = {309--324},
  publisher    = {Springer},
  year         = {2025},
  doi          = {10.1007/978-3-032-11127-2_25}
}

@article{DobrevFPS07,
  author       = {Stefan Dobrev and
                  Paola Flocchini and
                  Giuseppe Prencipe and
                  Nicola Santoro},
  title        = {Mobile Search for a Black Hole in an Anonymous Ring},
  journal      = {Algorithmica},
  volume       = {48},
  number       = {1},
  pages        = {67--90},
  year         = {2007},
  doi          = {10.1007/S00453-006-1232-Z}
}

@inproceedings{ChalopinDS07,
  author       = {J{\'{e}}r{\'{e}}mie Chalopin and
                  Shantanu Das and
                  Nicola Santoro},
  title        = {Rendezvous of Mobile Agents in Unknown Graphs with Faulty Links},
  booktitle    = {Procs. of the 21st International Symposium on Distributed Computing (DISC)},
  pages        = {108--122},
  publisher    = {Springer},
  year         = {2007},
  doi          = {10.1007/978-3-540-75142-7_11}
}

@article{DobrevFPS06,
  author       = {Stefan Dobrev and
                  Paola Flocchini and
                  Giuseppe Prencipe and
                  Nicola Santoro},
  title        = {Searching for a black hole in arbitrary networks: optimal mobile agents
                  protocols},
  journal      = {Distributed Computing},
  volume       = {19},
  number       = {1},
  pages        = {1--35},
  year         = {2006},
  doi          = {10.1007/S00446-006-0154-Y}
}

@article{Fraenkel70,
    author = {Aviezri  S. Fraenkel},
    title = {Economic Traversal of Labyrinths},
    journal = {Mathematics Magazine},
    volume = {43},
    number = {3},
    pages = {125--130},
    year = {1970},
    publisher = {Taylor \& Francis},
    doi = {10.1080/0025570X.1970.11976025}
}

    \appendix
    \section{Proofs}\label{appendix:proofs}

\subsection{Proofs for $\Aperp$}\label{sec:proofs_Aperp}
\begin{appendixlemma}{\ref{lemma:SEperp_safe}}
    If a port $\port{v,w}$ is marked as {\SAFE} at round $t$, then $w$ is truly a safe node and $\port{v,w}$ permanently remains marked as {\SAFE} thereafter.
\end{appendixlemma}
\begin{proof}
    According to $\Aperp$, a port $\port{v,w}$ is marked as {\SAFE} as soon as an agent located in the ingoing sub-port $\inport{v,w}$ is activated. 
    In fact, by construction, an agent can be located on $\inport{v,w}$ only if it has traversed the channel $(w,v)$: this proves the fact that $w$ is safe.
    According to our algorithm, this mark will never be updated. 
\end{proof}

\begin{appendixlemma}{\ref{lemma:SEperp_dangerous}}
    If a port $\port{v,w}$ is marked as {\DANGEROUS} at round $t$ and node $w$ is truly a black hole, then such a port will be permanently marked as {\DANGEROUS}.
\end{appendixlemma}
\begin{proof}
    According to $\Aperp$, the only case in which a port $\port{v,w}$ marked as {\DANGEROUS} is then marked with another value occurs when an agent enters node $v$ through the ingoing port $\inport{v,w}$. In this case, the port will be marked as {\SAFE}.
    However, since $w$ is a black hole, this case can never happen.
\end{proof}

\begin{appendixlemma}{\ref{lemma:SEperp_false_dangerous}}
    If a port $\port{v,w}$ is marked as {\DANGEROUS} at round $t$ and node $w$ is safe and $\{v,w\}$ is recurrent, then $\port{v,w}$ will be eventually and permanently marked as {\SAFE}.
\end{appendixlemma}
\begin{proof}
    According to $\Aperp$, if an agent $a$ marks $\port{v,w}$ as {\DANGEROUS} and $\{v,w\}$ is recurrent, then, eventually, $a$ will travel along the edge $\{v,w\}$. 
    As soon as $a$ enters $w$ through $\inport{w,v}$, it understands $w$ is safe.
    $\Aperp$ makes an agent mark as {\SAFE} any {\UNEXPLORED} or {\DANGEROUS} port from which it has just entered (except for the ports on the black holes).
    If $\port{w,v}$ was marked as {\UNEXPLORED}, then $a$ positions itself on $\outport{w,v}$ and returns to $v$ along the same (recurrent) edge: as soon as it will go back to $\port{v,w}$, $a$ will mark this port as {\SAFE}.
    However, if $\port{w,v}$ was already marked as {\SAFE} or it was marked as {\DANGEROUS}, then it means that at least one agent has already moved or is going to move along the channel $(w,v)$. 
    Such agents will be responsible for marking $\port{v,w}$ as {\SAFE}.
    By \Cref{lemma:SEperp_safe}, the mark {\SAFE} on $\port{v,w}$ is permanent.
\end{proof}

\begin{appendixlemma}{\ref{lemma:always_agent_travels}}
    Consider a team of at least $2\eta' + \sum_{b\in B}\delta(b) +1$ agents executing $\Aperp$ on a network $\tmpG = (G=(V,E), \Lambda, \varrho, B)$. Then, for each round $t$, there exists a time $t'\geq t$ where at least one agent travels along an edge.
\end{appendixlemma}
\begin{proof}
    Let $t\in \NatO$. 
    At time $t$, at most $\sum_{b\in B}\delta(b)$ agents may be trapped in the black holes or located in the outgoing dangerous sub-ports of transient dangerous edges, waiting indefinitely for their reappearance; moreover, at most $2\eta'$ agents may be located at the $2\eta'$ outgoing sub-ports of the transient safe edges, also waiting in vain for their reappearance.
    Note, in fact, that the rotor-router mechanism ensures that at most one agent may be located on an outgoing sub-port at each time.
    Thus, at least one agent, say $a$, is located in some safe node, say $v$, not on the outgoing sub-port of a transient edge.
    Note that, since $\tmpG$ is temporally connected, each safe node has at least one incident edge which is both safe and recurrent.
    According to $\Aperp$, if $a$ is activated and there is a candidate port (i.e., not marked as {\DANGEROUS} and not occupied by other agents), then it moves to the corresponding outgoing sub-port.
    Eventually, either $a$ or one agent that occupies the recurrent outgoing sub-ports of $v$ will be activated when the corresponding edge is present, and it will travel along it (thanks to the eventual transportation condition).
    Instead, suppose all the ports of $v$ are marked as {\DANGEROUS} at time $t$.
    Among them, at least one port, say $\port{v,w}$, is truly safe and recurrent; thus there is an agent on $\port{w,v}$ that will eventually come back to $v$ through the port $\port{v,w}$ and will mark it as {\SAFE}.
    In any case, at least one agent travels along an edge at a time $t'\geq t$.
\end{proof}

\begin{appendixlemma}{\ref{lemma:recurrent_edges_visited_inf}}
    If a safe node $v$ is visited by an incoming agent an infinite number of times, then each safe node $w$ connected to $v$ through a recurrent edge will be visited an infinite number of times.
\end{appendixlemma}
\begin{proof}
    When an agent enters $v$ through a port $\inport{v,w}$, after possibly marking it as {\SAFE}, it has two options: \emph{(i)} go back through $\outport{v,w}$ to mark $\port{w,v}$ as {\SAFE}, or \emph{(ii)} use the rotor-router mechanism to find the next candidate port to be explored.
    Yet, note that \emph{(i)} can be done at most one time for a given port (i.e., only if the agent finds the port still {\UNEXPLORED}).
    Thus, after a finite number of visits, agents must use the rotor-router mechanism to choose the next node to visit.
    So, assume we always apply \emph{(ii)}, where $\rotor_v$ simply cycles over all the ports: in this case, the agent selects the first candidate port (i.e., empty from agents and not marked as {\DANGEROUS}) starting from the port pointed by $\rotor_v$.
    Thus, all the {\SAFE} and {\UNEXPLORED} ports will be explored at least once.
    In particular, all the dangerous {\UNEXPLORED} ports will be marked as {\DANGEROUS} and no longer explored.
    On the contrary, all the safe and recurrent {\UNEXPLORED} ports will be marked as {\DANGEROUS} for the first exploration, and then marked as {\SAFE} as soon as an agent comes back from the same port.
    All the safe and transient {\UNEXPLORED} ports will be marked as {\DANGEROUS}, and possibly they will remain {\DANGEROUS} forever.
    Thus, eventually, all the safe and recurrent ports of $v$ will be marked as {\SAFE}, and they will be cyclically visited by the agents.   
\end{proof}

\begin{appendixthm}{\ref{th:SE_perpetual}}
	$\algo^\infty$ ensures a team of at least $2\eta' + \sum_{b\in B}\delta(b) +1$ oblivious and anonymous agents arranged on the anonymous nodes of a temporally connected network 
	$\tmpG$ under the {\ssynch}-{\TRAN} model to solve $\SEperp$ without any prior knowledge of $\tmpG$.
    Eventually, $\algo^\infty$ provides a highly reliable port marking on $\tmpG$.
\end{appendixthm}
\begin{proof}
	By \Cref{lemma:always_agent_travels}, it follows that there exists an agent $a$ that will never be trapped by a black hole, and continues to move along recurrent edges. 
    This means that there exists at least one node that is visited infinitely often by $a$.
	By \Cref{lemma:recurrent_edges_visited_inf}, if $a$ reaches $v$ infinitely often, then it will travel along the adjacent safe and recurrent edges infinitely often.
	Thus, all the safe and recurrent adjacent nodes of $v$ will be visited infinitely often. 
    Recursively, since $\tmpG$ is temporally connected, each safe node (and each safe recurrent edge) will be visited infinitely often.
    
    Since an agent marks a port as {\SAFE} as soon as it enters it, it follows that all the ports of safe and recurrent edges will be marked as {\SAFE}.
    Moreover, the rotor-router mechanism ensures all the dangerous edges incident to $v$ are permanently marked as {\DANGEROUS}; while all the transient edges incident to $v$ are permanently marked as either {\SAFE} (when truly safe) or as {\DANGEROUS}.
    These claims, together with \Cref{lemma:SEperp_safe,lemma:SEperp_dangerous,lemma:SEperp_false_dangerous}, prove that the final port marking is highly reliable.
\end{proof}

\subsection{Proofs for $\Afini$}\label{sec:proofs_Afini}

\begin{appendixlemma}{\ref{lemma:SEfini_Exploration_finite}}
    Consider a team of at least $2\eta' + \sum_{b\in B}\delta(b) +1$ agents executing $\Afini$ on a network $\tmpG = (G=(V,E), \Lambda, \varrho, B)$. Then, the Exploration phase is finite and ensures that all safe nodes are visited at least once and that at least one agent remains free.
\end{appendixlemma}
\begin{proof}
    The Exploration phase exploits $\Aperp$ as a subroutine and ends as soon as one agent has collected all $s=|V\setminus B|$ node IDs and has set its $\agentstatus$ to {\PROPAGATION}.
    Since it uses $\Aperp$, we know that at least one agent will never be trapped in a black hole.
    Hence, we need to prove that there exists at least one agent that will collect all the node IDs.
    An agent $a$ cumulatively collects node IDs in two possible ways: \emph{(i)} by directly visiting a not-yet-named node or \emph{(ii)} by visiting an already-visited node $v$ and merging its own $\agentnodelist_a$ with $\nodelist_v$.
    By contradiction, let us suppose that there will never be an agent that collects all the node IDs.
    As a consequence, the Exploration phase will never end: agents will indefinitely execute $\Aperp$, and all the safe nodes will be explored infinitely often.
    This means that there will be a partition of $V(G)=\{V_1,V_2\}$ and a partition of $\agents=\{\agents_1, \agents_2\}$ such that the sub-team $\agents_1$ ($\agents_2$, resp.) will be forever visiting only the nodes in $V_1$ ($V_2$, resp.).
    However, the induced subgraphs of $G[V_1]$ and $G[V_2]$ are connected by at least one recurrent edge in $\tmpG$; otherwise $\tmpG$ would not be temporally connected.
    By \Cref{lemma:recurrent_edges_visited_inf}, we know that all the safe recurrent edges will be traveled infinitely often by executing $\Aperp$.
    Yet, this contradicts the hypothesis that there cannot exist a node in $\tmpG$ that will be visited by both $\agents_1$ and $\agents_2$.    
\end{proof}

\begin{appendixlemma}{\ref{lemma:SEfini_Propagation_finite}}
    The Propagation phase is finite, ensures that all safe nodes have $\nodestatus$ marked as {\EXPLCOMPLETE} and that at least one agent remains free.
\end{appendixlemma}
\begin{proof}
    The Propagation phase starts as soon as an agent sets $\agentstatus\gets \PROPAGATION$ and ends as soon as an agent has $s$ node IDs in its $\agentnodelist$ set to {\true}, and thus all the $s$ safe nodes have been marked as {\EXPLCOMPLETE} in their whiteboard.
    By \Cref{lemma:SEfini_Exploration_finite}, we know that this phase eventually starts; we have to prove that this phase stops only after all safe nodes have been visited at least once after the Exploration phase.
    According to $\Afini$, in this phase the free agents with $\agentstatus=\PROPAGATION$ will navigate on the network still using $\Aperp$ to set the $\EXPLCOMPLETE$ value in the whiteboard of each safe node.
    Since the agents have used and still use $\Aperp$ for navigating the network, we know by \Cref{lemma:SEperp_safe} that a {\SAFE}-marked port is truly safe: this ensures that the final port marking is at least weakly reliable.
    
    The use of $\Aperp$ ensures that at least one agent (those with $\agentstatus= \PROPAGATION$) will never fall into a black hole.
    With the same argument as in \Cref{lemma:SEfini_Exploration_finite}, we can claim that the free agents will mark all the $s$ nodes as {\EXPLCOMPLETE}, and they will collect the list of all the $s$ nodes which have been marked.
\end{proof}

\begin{appendixlemma}{\ref{lemma:SEfini_Termination_finite}}
    The Termination phase is finite and ensures that all free agents, of which are at least one, are at the center of a safe node.
\end{appendixlemma}
\begin{proof}
    The Termination phase starts as soon as a free agent sets $\agentstatus\gets \TERMINATION$; by \Cref{lemma:SEfini_Propagation_finite}, we know that eventually this phase starts with all the safe nodes marked as {\EXPLCOMPLETE}.
    So, as soon as a free agent is activated, it sees that the node in which it is positioned is marked as {\EXPLCOMPLETE} and contains the list of all the $s$ safe nodes marked as {\true}; so, it understands that it has to move to the center of the node and terminates.
    No other action will be performed.
\end{proof}

\begin{appendixthm}{\ref{th:SE_finite}}
	$\Afini$ ensures a team of at least $2\eta' + \sum_{b\in B}\delta(b) +1$ agents with IDs arranged on the anonymous nodes of a temporally connected network 
	$\tmpG$ under the {\ssynch}-{\TRAN} model to solve $\SEfini$ with explicit termination, only knowing $s=|V\setminus B|$.
    Eventually, $\Afini$ provides a \emph{weakly reliable} port marking on $\tmpG$.
\end{appendixthm}
\begin{proof}
    By \Cref{lemma:SEfini_Exploration_finite}, we know that all the nodes are visited at least once in finite time and that at least one agent remains free.
    By \Cref{lemma:SEfini_Propagation_finite}, we know that the information about the exploration termination is \emph{explicitly} propagated in all the $s$ safe nodes by writing the value {\EXPLCOMPLETE} on their whiteboards; even in this phase, at least one agent remains free.
    During the propagation, the free agents use---and possibly continue to expand---a spanning subgraph $\tmpGSAFE$ of $\tmpG$, which guarantees that each {\SAFE}-marked port is truly safe; thus, the final port marking is weakly reliable.
    However, $\tmpGSAFE$ may not be temporally connected (as proved in \Cref{th:impossible_reliable_PM}, some transient edges of $\tmpGSAFE$ may disappear after the propagation phase and leave some safe nodes isolated).
    By \Cref{lemma:SEfini_Termination_finite}, we know that eventually all the free agents will move permanently to the center of some safe nodes and set $\agentstatus$ to {\TERMINATION}.
    This establishes the explicit termination.
\end{proof}

\section{Pseudocodes, tables and figures}\label{appendix:pseudocodes}
\Cref{algo:LCM} shows the Look-Compute-Move cycle executed by each agent.
\Cref{algo:SE_perpetual} and \Cref{algo:SE_finite} present the pseudocodes of the algorithm $\Aperp$ for $\SEperp$ (explained in \Cref{sec:SE_perpetual}) and of the algorithm $\Afini$ solving $\SEfini$ (explained in \Cref{sec:SE_finite}), respectively.
\Cref{table:variables_SEfini} lists the variables used by $\Aperp$ and $\Afini$.

\begin{table}[ht]
    \centering
    \footnotesize{
    \renewcommand{\arraystretch}{1.3}{
    \begin{tabular}{|c|l|l|}
        \hline
        Memories                            & Variables & Usages\\
        \hline\hline
        \multirow{4}{*}{\bf $\wboard_v$}   
                                            & \cellcolor{gray!20}{$\ports_v$} & \cellcolor{gray!20}{Array in the form $\{\UNEXPLORED,\SAFE,  \DANGEROUS\}^{\delta(v)}$}\\
                                            &  \cellcolor{gray!20}{$\rotor_v$} & \cellcolor{gray!20}{Pointer containing an integer in $[0,\delta(v)-1]$ } \\
                                            &  $\nodeid_v$ & Contains a pair $\underline{x}=(a,j)$ with $j\in \NatO$ \\
                                            &  $\nodestatus_v$ & Blank or {\EXPLCOMPLETE} \\
                                            &  $\nodelist_v$ & List of entries in the form $(\underline{x},\texttt{bool})$ where \texttt{bool} is a boolean \\
        \hline
        \multirow{3}{*}{\bf $\nbook_a$}     & $\agentstatus_a$  & Contains a value in $\{\EXPLORATION, \PROPAGATION, \TERMINATION\}$\\
                                            & $\agentnodelist_a$  & As $\nodelist_v$\\
                                            & $\counter_a$  & Contains an integer\\
        \hline
    \end{tabular}
    }}
    \caption{Variables used in $\Afini$. The gray-colored variables are also adopted in $\Aperp$.}
    \label{table:variables_SEfini}
\end{table}

\begin{algorithm}[th]
\caption{LCM cycles executed by the agents in $\AS(t)$.}
\label{algo:LCM}
\small

    $a_{i_1},\dots, a_{i_{h(t)}} \gets \AS(t)$\;

        $a \gets a_{i_j}$\;
        \tcc{Executed in parallel by all the agents in $\AS(t)$}
        \MyBlock{Look}{
            $\sigma\gets \langle pos(a), \nbook(a)\rangle$ \;
    	}
        \tcc{Executed in sequence by all the agents in $\AS(t)$}
        \ForEach{$a_{i_h}$ in $\AS(t)$}{

        \MyBlock{Compute}{
            $\sigma\gets \sigma + \langle \wboard_v, \phi_v \rangle$ \;
    		$(p, \chi) \gets \algo(\sigma)$\;
    		\If{$\chi \neq \NaN$}{
    			$\wboard_v \gets \chi$\;
    		}
            \eIf{$p \in [0,\delta(v)-1]$}{
    			
    				$a$ moves to the outgoing port $p$\;
    			
    		}{
            \If {$p =\cnode$}{
    				$a$ moves to the center of $v$\;
    		}
            }
        }
    }
    \tcc{Executed in parallel by all the agents in $\AS(t)$}
    \MyBlock{Move}{
        \If{$a$ is located on an outgoing port $p$}{
			$e \gets $ edge at port $p$\;
            \If{$\varrho(e,t) =1$}{
				$a$ moves along $e$ and reaches the other node\;
			}
        }
    }
\end{algorithm}
\begin{figure}[ht]
    \centering
     \begin{tikzpicture}[scale=0.7,node distance=4cm, on grid,transform shape, align=center, state/.style={ellipse, draw, minimum width=3.2cm,
        minimum height=1.4cm}]   
          \node[state,initial, fill=yellow!20]  (U)   {$\UNEXPLORED$};
          \node[state,fill=red!20]   (D)  [below left=of U]  {$\DANGEROUS$};
          \node[state,fill=green!20]  (S)  [below right=of U]  {$\SAFE$};
        
          \path[->, shorten >=0.05cm] (U)     edge [bend right] node [left]           {$a$ tries the ingoing port}  (D)
                            edge [bend left] node [right]                    {$a$ reaches the outgoing port}  (S)      
                    (D)     edge [left] node [above]          {$a$ comes back}  (S)
                            edge [loop left]  node {$a$ is trapped in a black hole}  (D)
                    (S) edge [loop right] node {} (S);
        \end{tikzpicture}
    \caption{Diagramm for the Cautious Walk technique.}
    \label{fig:DFA_cautious_walk}
\end{figure}
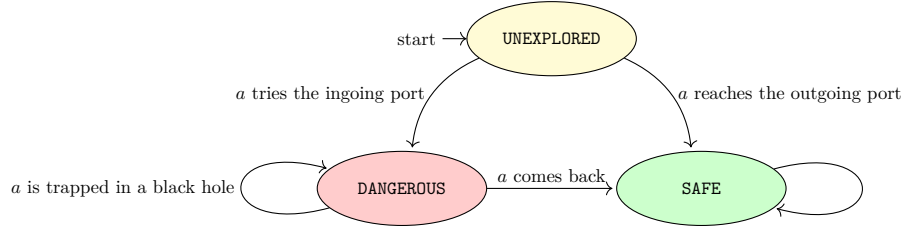

\begin{minipage}[t]{0.42\textwidth}
    \begin{algorithm}[H]
    \footnotesize
      \SetKwFunction{isempty}{is\_empty}
    	\SetKwProg{Fn}{Function}{ is}{end}
    	\Fn{\isempty{x: int} : int}{ 
    		$w \gets \lambda^{-1}_v(x)$\;
    		\If{$\exists a\in \agents \text{ on } \inport{v,w} \text{ or } \outport{v,w}$}{
    			\KwRet 0\;
    		}
    		\KwRet 1\;
    	}
        \caption{Utility to check if a port is empty.}
    \end{algorithm}  
\end{minipage}
\hfill
\begin{minipage}[t]{0.56\textwidth}
    \begin{algorithm}[H]
    \footnotesize
      \SetKwFunction{nextcandidate}{next\_candidate}
    	\SetKwProg{Fn}{Function}{ is}{end}
    	\Fn{\nextcandidate{} : int}{ 
    		\ForEach{$i\in [0, \delta_v-1]$}{
    			$j \gets (\rotor_v + i) \mod \delta_v$\;
    			\If{$\ports_v[j] \neq \DANGEROUS$ and \isempty{j}}{
    				\KwRet $j$\;
    			}
    		}
    		\KwRet \NaN\;
    	}
        \caption{Utility for computing the next candidate port to traverse.}
        \label{algo:next_candidate}
    \end{algorithm} 
\end{minipage}
\begin{algorithm}[th]
    \caption{Algorithm $\Aperp$ solving $\SEperp$.}
    \label{algo:SE_perpetual}
    \small
    \tcp{{Compute} step of an agent $a$ activated at a node $v$.}
    $pos(a)\gets$ relative position of $a$ within $v$\;
        \uIf{$pos(a)= \cnode$}{
            $j\gets \nextcandidate{} $\;
             \If{$j \neq \NaN$}{
                $\rotor_v \gets j+1 \mod \delta_v$\;
                $w\gets \lambda^{-1}(j)$\;
                $a$ moves to $\outport{v,w}$\;
                \If{$\ports_v[j] = \UNEXPLORED$}{
                    $\ports_v[j] \gets \DANGEROUS$\;
                }
             }
        }\uElseIf{$pos(a) = \inport{v,w}$}{
            \eIf{$\ports_v[\lambda_v(w)] = \SAFE$}{
                $a$ moves to $\cnode$\;
            }{ 
                \eIf{$\ports_v[\lambda_v(w)] = \UNEXPLORED$}{
                    $a$ moves to $\outport{v,w}$\;
                }{ 
                    $a$ moves to $\cnode$\;
                }
                $\ports_v[\lambda_v(w)] \gets \SAFE$\;
            }
        }\Else{
        \tcp{If $pos(a) = \outport{v,w}$, $a$ does nothing.}
        } 
    
\end{algorithm} 
\begin{algorithm}[th]
    \caption{Algortihm $\Afini$ solving $\SEfini$.}
    \label{algo:SE_finite}
    \small
    \tcp{{Compute} step of an agent $a$ activated at a node $v$.}
	$s,a$ are the parameters of the algorithm\;
    \Switch{$\nodestatus_v$}{
        \Case{$\EXPLORATION$}{
            \eIf{$\nodestatus_v \neq \EXPLCOMPLETE$}{
                \If{$\nodeid_v$ is $\NaN$}{
                    $j \gets \counter_v$\;
                    $\nodeid_v \gets (a,j)$\;
                    $\counter_v++$\;
                    $\agentnodelist_a \gets \agentnodelist_a \cup \{(a,j,\false)\}$\;
                }
                $\ell \gets \nodelist_v \cup \agentnodelist_a$\;
                $\agentnodelist_a \gets \ell$\;
                $\nodelist \gets \ell$\;
                \eIf{$len(\ell) = s$}{
                    $\agentstatus_a \gets \PROPAGATION$\;

                }{ 
                    executes $\Aperp$\;
                }
            }{ 
                $\agentstatus_a \gets \PROPAGATION$\;
            }
            \Return\;
        }
        \Case{$\PROPAGATION$}{
            $\nodestatus_v \gets \EXPLCOMPLETE$\;
            $nid \gets \nodeid_v$\;
            sets $\true$ the entry of $nid$ in $\agentnodelist_a$\;
            $\ell \gets []$\;
            \ForEach{$(\underline{x}, bool) \in \agentnodelist_a$}{
                \eIf{$\exists (\underline{x}, bool') \in \nodelist_v$}{
                    $\ell \gets \ell.append((\underline{x}, bool \lor bool'))$\;
                }{ 
                    $\ell \gets \ell.append((\underline{x}, bool))$\;
                }
            }
            $\agentnodelist_a\gets \ell$\;
            $\nodelist_a\gets \ell$\;
            
            \eIf{$\agentnodelist_a$ has $s$ {\true} entries}{
                $\agentstatus_a\gets \TERMINATION$\;
            }{
                executes $\Aperp$\;
            }
            \Return\;
        }
        \Case{$\TERMINATION$}{
            $a$ moves to $\cnode$ of $v$\;
            \Return\;
        }
    }   
\end{algorithm}

\end{document}